\documentclass[11pt]{article} 
\usepackage{amssymb,latexsym,amsthm,amsopn,amstext,amscd,mathrsfs,amsmath,graphicx}
\usepackage{layout} 
\usepackage[all]{xy}

\theoremstyle{definition} 
\newtheorem{dfn}{Definition}

\theoremstyle{plain}
\newtheorem{tma}{Theorem}
 
\newtheorem{lma}{Lemma}

\newtheorem{cor}{Corollary}

\theoremstyle{remark}

\newcommand{\gf}{\varphi} 
 
\newcommand{\geps}{\varepsilon}

\DeclareMathOperator{\pr}{pr} 
\DeclareMathOperator{\enc}{emb} 
\DeclareMathOperator{\id}{id}
\DeclareMathOperator{\ba}{bar}

\newcommand{\ti}[1]{\stackrel{\thicksim}{#1}} %\sptilde
\newcommand{\bl}[1]{\mathcal{L}_{#1, \star}}
\newcommand{\h}[2]{\hom(#1 , #2)}

\newcommand{\AG}{\mathcal{A}/\mathcal{G}_\star}
\newcommand{\AGp}{\mathcal{A}/\mathcal{G}_{\star,p}}

\title{
Effective theories of connections and curvature: \\
abelian case}
\author{
Homero G. D\'\i az-Mar\'\i n$^{1, 2}$%
\footnote{e-mail: \ttfamily homero@matmor.unam.mx} 
 and Jos\'e A. Zapata$^1$%
\footnote{e-mail: \ttfamily zapata@matmor.unam.mx} 
\\%%for iop this line is commented%%
{\it $^1$Instituto de Matem\'aticas, Universidad Nacional Aut\'onoma de M\'exico} \\ 
{\it A.P. 61-3, Morelia Mich. 58089, M\'exico}\\
{\it $^2$Instituto de F\'\i sica y Matem\'aticas,}\\ 
{\it Universidad Michoacana de San Nicol\'as de Hidalgo,}\\
{\it Edif. C-3, C. U., Morelia, Mich.  58040, M\'exico}
}

\date{}
\begin{document}

\maketitle

\begin{abstract} 
We introduce a notion of measuring scales for quantum abelian gauge systems. At each measuring scale a finite dimensional affine space stores information about the evaluation of the curvature on a discrete family of surfaces. 
Affine maps from the spaces assigned to finer scales to those assigned to coarser scales play the role of coarse graining maps. 
This structure induces a continuum limit space which contains information regarding curvature evaluation on all piecewise linear surfaces with boundary. 
The evaluation of holonomies along loops is also encoded in the spaces introduced here; thus, our framework is closely related to loop quantization and it 
allows us to discuss 
effective theories in a sensible way. 
We develop basic elements of measure theory on the introduced spaces which are essential for the applicability of the framework to the construction of quantum abelian gauge theories. 
\end{abstract}

%%%%%%%%%%%%%%%%%%%%%%%%%%%%%%%%%%%%%%%%%%%%%%%%%%%%%%%%%%%%%%%%%%%%%%%%%%%

\section{Introduction}
An effective theory has the objective of providing an approximate description of the behavior of the aspects of the system that can be measured at the given scale. We introduce a model of measuring scales for abelian gauge fields (connections on a $U(1)$ principal bundle) and provide the basic structures that makes it useful in building a quantum theory. The differential geometric setting that describes in detail the classical counterpart of this scenario was described in \cite{DZ}.

Our construction is closely related to loop quantization in the canonical form and in its covariant (spin foam) form, but we have   
added extra structure. The aim is to provide a language for discussing effective theories within loop quantized theories. A configuration of the effective theory may identify states of the system which can be resolved (differentiated) when measuring at finer scales, but it should be able to distinguish macroscopically different configurations.  In this respect, holonomies should not be the only variables measuring the connection. The following simple example captures one of the key ideas behind the structures introduced in this paper and is explained in detail in \cite{DZ}. Consider a one-parameter family of connection 1-forms in three dimensional euclidean space 
$A_\lambda = \lambda x  {\rm d} y$ and their respective magnetic fields $B_\lambda = \lambda {\rm d} x \wedge {\rm d} y$. If we know that the holonomy around one given loop is trivial we cannot 
determine the value of the parameter $\lambda$; 
we cannot say if the measured field is weak or arbitrarily strong. However, if we  measure the magnetic flux through a surface that has the loop as its boundary (and which is transverse to the $z$ direction everywhere), we can certainly determine 
the value of $\lambda$. 

Our framework can be used as a tool to study the macroscopic behavior of loop quantized abelian gauge theories, or it can be used as a variant of the kinematics of loop quantization to construct a quantum abelian gauge theory in the continuum. Most of this paper deals with the kinematical part of the proposal, and in the last section we give a short description of 
how the developed structures fit into the general plan to construct physical quantum gauge field theories. 

Curvature can be defined in our framework. The dynamics of a quantum gauge theory is usually constructed from the quantization of a classical expression which involves the curvature of the connection. Many classical observables (of local and global character) are written in terms of the curvature; the existence of their quantum counterparts certainly deserves to be studied. 
Also, curvature observables yield information about the topology of the bundle on which the gauge theory is defined. The holonomy variables of a discrete set of loops do not capture such information. 

In Section \ref{sec2} we introduce a notion of measuring scales 
associated with a triangulation 
$\Delta$, and recall 
a notion of simplified configurations called $\Delta$-flat connections which let us evaluate holonomies along loops economically. The space of possible holonomy evaluations is the finite dimensional space 
$\mathcal{A}/\mathcal{G}_{\Delta, \star}^\mathrm{Hol}=\h{\bl{\Delta}}{G}$; see \cite{MMZ} and \cite{DZ}. 
In Section \ref{sec4}, working in the case $G=U(1)$, we define the affine space $\Omega_{\Delta,p}$ that stores the surface evaluations of the curvature at scale $\Delta$; furthermore we show how a covering map
$$
\mathfrak{e}_\Delta : 
\Omega_{\Delta,p}\times\ \h{\pi_1(M,\star)}{U(1)} 
\rightarrow\mathcal{A}/\mathcal{G}_{\Delta, \star}^\mathrm{Hol}
$$
let us recover the information about 
holonomy evaluation when we know the curvature evaluation 
and the holonomy evaluation along a finite set of generating 
noncontractible loops. 
While the factor $\Omega_{\Delta,p}$ yields the holonomy evaluation on contractible loops, the factor $\h{\pi_1(M,\star)}{U(1)}$ stores the holonomy for non contractible loops. Related to this map, in \cite{DZ} we showed that for $M$ homeomorphic to the sphere $S^d$, each element $\omega\in\Omega_{\Delta,p}$ determines a unique bundle over $M$ and 
determines the connection modulo gauge up to ``microscopic details'' in the sense that any smooth connection compatible with the data $\omega$ can be locally deformed into 
a singular connection determined only by $\omega$. 
In Section \ref{sec3} we define coarse graining maps; each coarse graining map 
produces data at a coarser scale from data on a finer scale. We also 
define an affine space which stores curvature evaluations at ``the finest scale'' of a sequence of scales, 
$\underleftarrow{\Omega}_p = Projlim(\Omega_{i,p})$; from this space we can coarse grain to know the curvature evaluation at any scale of the sequence. 
Section \ref{sec5} introduces a family of kinematical gaussian measures on $\underleftarrow{\Omega}_p$ depending on the area form of a riemannian metric on 
$M$. 
Section \ref{convergence} proves certain convergence results which allow our kinematical structures at the 
continuum limit deal with geometric objects of the piecewise linear (PL) category. 
Section \ref{sec6} studies curvature functions and some of their properties. 
Section \ref{indep} 
proves 
that our results are, in a specific sense, independent of the choices made during the construction. 
In the last section we sketch how the structures introduced in this paper fit into the plan to construct quantum field theories following Wilson's renormalization procedure. 

%%%%%%%%%%%%%%%%%%%%%%%%%%%%%%%%%%%%%%%%%%%%%%%%%%%%%%%%%%%%%%%%%%%%%%%%%%%%%%%

\section{Holonomy evaluations for effective theories} \label{sec2}

Let $M$ be a smooth compact oriented $d$ dimensional manifold, and let $(E,p,M)$ be a principal $G$-bundle with 
total space $E$, base space $M$, 
projection  $p:E\rightarrow M$ and compact structure group $G \approx p^{-1}(x)$ for all $x \in M$. 
Let $\mathcal{A}_p$ denote the space of smooth connections on $E$, 
$\mathcal{G}_p$ denote the group of gauge transformations, and 
$\mathcal{G}_{\star,p} \subset \mathcal{G}_p$ denote the subgroup that acts as the identity on the fiber $p^{-1}(\star)$ over a base point $\star\in M$. 
In order to identify the fiber over the base point with $G$, we fix 
a point $y_\star$ over the fiber $p^{-1}(\star)$. 
We declare two piecewise smooth closed paths based on $\star\in M$, $l,l':[0,1]\rightarrow M$,  as honomicaly equivalent if $\mathrm{Hol}_A(l)=\mathrm{Hol}_A(l')$ for all $A\in\mathcal{A}$. We can describe the holonomy $\mathrm{Hol}_A(l)$ by an element of $G$ that depends only on the $\mathcal{G}_{\star,p}$-class $[A]\in\mathcal{A}/\mathcal{G}_{\star,p}$.

Recall that each class $[A]$ can be described in terms of the set of its holonomies along all piecewise smooth loops \cite{Ba}.

We say that two piecewise smooth loops modulo reparametrization $l,l'$ are \emph{retracing equivalent} if there is a finite sequence of piecewise smooth loops $l_0=l,l_1,\dots,l_n=l',$ such that for two consecutive loops $l_k,l_{k+1}$ we have $l_k=m_1\cdot m_2$ and $l_{k+1}=m_1\cdot m\cdot m^{-1}\cdot m_2$. Two retracing equivalent loops are also holonomically equivalent. We denote the retracing equivalence relation by $l\sim l'$. The set of equivalence classes of loops modulo retracing is a group.

We will work on the piecewise linear (PL) category. There are advantages of working with the PL category, for instance two holonomicaly equivalent PL loops are also retracing equivalent.

\begin{dfn}
Consider a simplicial complex $|\Delta|$ homeomorphic to $M$.
Let $\phi:|\Delta|\rightarrow M$ be a homeomorphism whose restriction on each simplex is a smooth non singular immersion.
We say that $(|\Delta|,\phi)$ is a smooth triangulation or  Whitehead triangulation of $M$. 
\end{dfn}

A Whitehead triangulation defines a PL structure that is compatible with the smooth structure of $M$ in the sense that a map that is piecewise linear is also piecewise smooth.
It is known that every smooth manifold admits smooth triangulations and that any two such triangulations are equivalent in the sense that the simplicial manifolds $|\Delta | , |\Delta' |$ are related by a piecewise linear map \cite{KS}.  

Take a fixed vertex $\star\in\phi({\rm Sd}\,|\Delta|^{(0)})$, and let $\bl{M}$ be the group of PL loops modulo retracing. The holonomy map $\mathrm{Hol}_A$ can be regarded as a homomorphism $\mathrm{Hol}_A:\bl{M}\rightarrow G$. Furthermore the map $\mathrm{Hol}([A]):=\mathrm{Hol}_A$ defines an inclusion
$$
	\mathrm{Hol}:\AGp\rightarrow\h{\bl{M}}{G}.
$$

One verifies that this inclusion is dense in the Tychonoff topology. If we denote the closure of the inclusion $\AGp\subset\h{\bl{M}}{G} $ as $\overline{\AG}$, then we have the homeomorphism
$$\overline{\AG}\simeq\h{\bl{M}}{G}.$$
The uniform measure on the compact space $\overline{\AG}$ is called the Ashtekar-Lewandowski measure. This is the configuration space for loop quantization \cite{Th, Ro, AL}; its PL version was introduced in \cite{Za}. Each class $[A]\in\AGp$ can be described in terms of its holonomies along the huge group of all PL-loops $\bl{M}$.

We will specify the information available at a given scale in terms of a smooth triangulation $(|\Delta|,\phi)$. We use the triangulation  to define an equivalence relation of PL loops $l,l'$. 

\begin{dfn}
A point $p \in M$ is assigned the simplex of smallest dimension of $(|\Delta|,\phi)$ which contains it. 
In this way as a curve $c$ is traversed, it induces a sequence of neighboring simplices in $(|\Delta|,\phi)$.

We say that the loops $l,l'$ are equivalent at scale $(|\Delta|,\phi)$, $l\sim_{\Delta}l'$ if they induce the same sequence of neighboring simplices. 
The retracing equivalence relation $\sim$ for loops yields a retracing equivalence relation $\asymp$ for sequences of simplices. 
We denote the resulting finitely generated group of classes as
\begin{equation}\label{eqn:p_Delta}
	\bl{\Delta}:=\left(\left\{l\,\mid\,l\text{ is a PL loop}\right\}/\sim_\Delta\right)/\asymp.
\end{equation}
We call the elements of the set
\begin{equation}
	\mathcal{A}/\mathcal{G}_{\Delta, \star}^\mathrm{Hol}:=\h{\bl{\Delta}}{G}
\end{equation}
\emph{holonomy evaluations of $\Delta$-flat connections} on $M$.
\end{dfn}

Until this stage we have considered holonomy evaluations for each scale $(|\Delta|,\phi)$ which is only a part of the information that $\Delta$-flat connections encode. The rest of the information contained in $\Delta$-flat connections will be described starting in the next section; a more detailed account of the differential geometric aspects of $\Delta$-flat connections is given in \cite{DZ, MMZ}.

One advantage of working at scale $(|\Delta|,\phi)$ is that the group of classes of loops to be considered $\mathcal{L}_{\Delta,\star}$ is finitely generated. Hence $\h{\bl{\Delta}}{G}$ is a manifold.

\begin{dfn}\label{Sd}
We construct a simplicial complex ${\rm Sd}\,|\Delta|$ called the baricentric subdivision of $|\Delta|$. 
Refine each simplex of $|\Delta |$ using the triangulation generated by adding a new vertex placed in its baricenter (according to the euclidean metric of the simplex). 
A subset of $n$ such vertices defines a $n-1$ dimensional simplex of ${\rm Sd}|\Delta |$ 
if the corresponding subset of simplices of $| \Delta |$ can be ordered by inclusion; 
in this case, it defines a geometric simplex contained in the higher dimensional simplex of the set. 
We denote the set of $k$-simplices on ${\rm Sd}\,|\Delta|$ as ${\rm Sd}\,|\Delta|^k,$ and the $k$-skeleton of this polyhedra as ${\rm Sd}\,|\Delta|^{(k)}=\cup_{l=0}^k{\rm Sd}\,|\Delta|^l\subset {\rm Sd}\,|\Delta|$. For the set of images on $M$ under $\phi$ of $k$-simplices of $ {\rm Sd}\,|\Delta|$ we will use the notation $({\rm Sd}\,|\Delta|^k,\phi)$, and for the image of the $k$-skeleton we will use 
$\phi({\rm Sd}\,|\Delta|^{(k)})$. 
\end{dfn}

For a fixed scale $(|\Delta|,\phi)$ we consider the set $\mathcal{L}_{\phi({\rm Sd}\,|\Delta|^{(1)}),\star}$ of retracing classes of loops in $\phi({\rm Sd}\,|\Delta|^{(1)})$. 
For every $l_\Delta \in\mathcal{L}_{\Delta, \star}$ we have a representative loop $l_1:[0,1]\rightarrow M$, contained in $\phi({\rm Sd}\,|\Delta|^{(1)})$ such that 
$l_\Delta=[[l_1]_\Delta]_\asymp$. By observing that a homotopy class in the $1$-skeleton $\phi({\rm Sd}\,|\Delta|^{(1)})$ corresponds to a retracing class, we can prove the following lemma.

\begin{lma}\label{lma:0}
The loop class group $\mathcal{L}_{\Delta, \star}$ is isomorphic to the group of retracing classes of loops in $\phi({\rm Sd}\,|\Delta|^{(1)})$, i.e. $\mathcal{L}_{\phi({\rm Sd}\,|\Delta|^{(1)}),\star}$. We have the natural isomorphism
$$
	\bl{\Delta}\simeq \pi_1\left(\phi({\rm Sd}\,|\Delta|^{(1)}),\star \right).
$$

\end{lma}

Since we have an embedding $\phi:{\rm Sd}\,|\Delta|^{(1)}\rightarrow M$, this lemma allows us  to give a holonomy evaluation at each scale for smooth connection classes $[A]\in\AGp, $ which can be expressed as the map
$$
	{\mathrm{Hol}}_\Delta:\AGp\rightarrow\mathcal{A}/\mathcal{G}_{\Delta,\star}^{\mathrm{Hol}}.
$$
This map assigns a continuum to one. A complete description of a connection $[A]\in\AGp$ will require a sequence of ever finer scales.

There is another ambiguity when we only consider holonomies to study connections at a given scale: the group of holonomy evaluations $\mathcal{A}/\mathcal{G}_{\Delta,\star}^\mathrm{Hol}$ loses track of the topology of the principal bundle $p:E\rightarrow M$, as we noticed on \cite{DZ}. 
Recovering the topological data of the bundle using only holonomies is possible only if we consider 
all the PL loops embedded in the manifold. 

On the other hand, if we consider evaluations of the curvature on simplicial surfaces at a given scale, 
as we do starting in the following section, we can recover the principal bundle structure; see \cite{DZ}.

Let $[l]\in\bl{M}$ be the retracing class of a PL loop $l$; it can be shown that the class $[[l]_\Delta]_\asymp\in\bl{\Delta}$ does not depend on the representative element  $l'\in[l]$. Therefore we can define the surjective group homomorphism
\begin{equation}\label{eqn:pr_Delta}
	\widehat{\pr_{\Delta}}:\bl{M}\rightarrow\bl{\Delta},
\end{equation}
which, when considering homomorphisms onto $G$, gives rise to
\begin{equation}\label{eqn:iota_Delta}
	\widehat{\iota_{\Delta}}:
	\mathcal{A}/\mathcal{G}_{\Delta, \star}^\mathrm{Hol}	\rightarrow		\overline{\AG}.
\end{equation}
At each scale the map $\widehat{\iota_\Delta}$ can be used to regularize certain functions from the continuum; see \cite{DZ}. In addition, given a sequence of scales 
$\{ (|\Delta_i|,\phi_i)\} $ and a sequence of measures $\{ \mu_{\Delta_i}\}$ we get a sequence of measures $\{ \widehat{i_{\Delta_i}}_*\mu_{\Delta_i}\}$ in $\overline{\AG}$ which may converge in some sense as the scale gets finer. These issues will be addressed starting in Section \ref{sec3}.

Now we describe the ingredients for the \emph{coarse graining} map. 
 As we mentioned in Lemma \ref{lma:0}, for each retracing class $\tilde{l}\in\bl{\Delta}$ there exists a simplicial loop of the form $\phi\circ l$ representing it, where $l\in\mathcal{L}_{{\rm Sd}\,|\Delta|^{(1)},\star}$, i.e. $l:[0,1]\rightarrow {\rm Sd}\,|\Delta|^{(1)}$, $\phi\circ l(0)=\phi\circ l(1)=\star,$ $\tilde{l}=\left[[\phi\circ l]_\Delta\right]_{\asymp}$. Take the loop $\phi\circ l:[0,1]\rightarrow M$, and take its retracing class $[\phi\circ l]$. Since $[\phi\circ l]$ does not depend on the representative $l$ of the retracing class $[l]$ in ${\rm Sd}\,|\Delta|^{(1)}$, this procedure defines a homomorphism
\begin{equation}\label{eqn:ba_Delta}
\ba_\Delta :\bl{\Delta}\rightarrow\bl{M},
\end{equation}
by $\ba_{\Delta}(\tilde{l})=[\phi\circ l]$. Thus we also have the submersion $\beta_\Delta:=\h{\ba_\Delta}{G},$
$$
	\beta_\Delta:\overline{\AG}\rightarrow\mathcal{A}/\mathcal{G}_{\Delta, \star}^\mathrm{Hol}.
$$
Since $\ba_\Delta$ is a section of $\widehat{\pr_\Delta}$, it follows that $\widehat{\iota_\Delta}$ is a section of $\beta_\Delta$, i.e.
$$
\widehat{\pr_\Delta}\circ\ba_\Delta=\id_{\bl{\Delta}},\qquad\beta_\Delta\circ \widehat{\iota_\Delta}=\id_{\h{\bl{\Delta}}{G}}.
$$
We thus obtain the following commutative diagram
$$\xymatrix{
		\AGp
			\ar@{^{(}->}[r]^{\mathrm{Hol}}\ar[rd]_{\mathrm{Hol}_\Delta}
	&
		\overline{\AG}
			\ar@/^10pt/[d]^{\beta_\Delta}
	\\
	&
		\mathcal{A}/\mathcal{G}_{\Delta, \star}^\mathrm{Hol}
}.$$

%%%%%%%%%%%%%%%%%%%%

In the following lines we show how to conveniently decompose the set of holonomy evaluations $\mathcal{A}/\mathcal{G}_{\Delta, \star}^\mathrm{Hol}$; 
roughly speaking the division will have a factor that encodes the holonomies along 
contractible loops and 
another factor encoding the holonomies along a subgroup of loops generating 
the fundamental group $\pi_1(M,\star).$ This is shown in Lemma \ref{lma:decomposition}.

The inclusion map $i : \left(\phi({\rm Sd}\,|\Delta|^{(1)}),\star \right) \to \left(M,\star \right)$ induces the epimorphism 
$i_\ast:\pi_1 \left(\phi({\rm Sd}\,|\Delta|^{(1)}),\star \right) \rightarrow\pi_1\left(M,\star\right)$. 
Hence there is an epimorphism 
$$
i_\ast:\bl{\Delta}\rightarrow \pi_1\left(M,\star\right).
$$ 
\begin{dfn}\label{L^0}
The subgroup of contractible loops will be denoted by 
$$
\bl{\Delta}^0=\ker i_\ast .
$$
\end{dfn}
$\bl{\Delta}^0$ is a normal subgroup of $\bl{\Delta}$, and since $i_\ast$ is an epimorphism we have 
$$
\pi_1\left(M,\star\right) \simeq \bl{\Delta} / \bl{\Delta}^0 . 
$$

Recall that Lemma \ref{lma:0} shows us that the groups 
$\bl{\Delta}$, 
$\mathcal{L}_{\phi({\rm Sd}\,|\Delta|^{(1)}),\star}$ and 
$\pi_1 \left(\phi({\rm Sd}\,|\Delta|^{(1)}),\star \right)$ 
are naturally isomorphic. 
Also recall that in simplicial homology $Z_1$ denotes $1$-chains (from the simplicial complex $\left(\phi({\rm Sd}\,|\Delta|),\star \right)$) that are cycles, $B_1$ denotes $1$-chains that are boundaries, and $H_1(M) = Z_1 / B_1$. 
Below we will use the ``chain map'' 
${\rm C} : \bl{\Delta}  \to Z_1$. 
%
% I cut here

%%%%%%

Let $\bl{\Delta}^A \subset \bl{\Delta}$ be a subgroup such that $\bl{\Delta} = \bl{\Delta}^0 \bl{\Delta}^A$; that is, that every loop  can be written as 
$l = l^0 l^A$ with $l^0 \in \bl{\Delta}^0$ and $l^A \in \bl{\Delta}^A$. Notice tat this implies that the set of homotopy classes of $\bl{\Delta}^A$ generates 
$\pi_1(M,\star)$. 

We will divide our presentation of the factorization of $\mathcal{A}/\mathcal{G}_{\Delta, \star}^\mathrm{Hol}$ in two cases. 
Firstly, we will consider the case in which 
$\bl{\Delta}^0 \cap \bl{\Delta}^A = \{ \id \}$. In this case there is a clean factorization of the connection's degrees of freedom into local and global; in the following sections this case allows for a faster study of the coarse graining process and the continuum limit. 
The core of the simplification comes from the fact that every loop in $\bl{\Delta}$ can be written as a product of a loop in $\bl{\Delta}^A$ and a contractible loop in a unique way; $\bl{\Delta}$ is the semidirect product of 
$\bl{\Delta}^A$ acting on $\bl{\Delta}^0$, which is denoted by 
$\bl{\Delta}^0 \rtimes \bl{\Delta}^A$. 
Secondly, we will consider the case in which there is no subgroup $\bl{\Delta}^A \subset \bl{\Delta}$ such that 
$\bl{\Delta} = \bl{\Delta}^0 \bl{\Delta}^A$ and $\bl{\Delta}^0 \cap \bl{\Delta}^A$ be trivial; in this case every loop can be factorized, but the factorization is not unique. 
The non trivial behavior in this case comes from the existence of non contractible loops $l \in \bl{\Delta}^A$ such that $C(l^r) \in B_1$ for some integer $r$. 
Let us explain with more detail. 
Notice that 
since the gauge group is abelian, 
in $\h{\bl{\Delta}^A}{U(1)}$ the group $\bl{\Delta}^A$ may be replaced by the abelian group $C(\bl{\Delta}^A)$; that is, there is a natural isomorphism 
$\h{\bl{\Delta}^A}{U(1)} \simeq \h{C(\bl{\Delta}^A)}{U(1)} 
$. 
Also notice 
that for any loop $l \in \bl{\Delta}$, the homology type of its associated $1$-chain 
$[{\rm C} (l) ]_{H_1} \in H_1(M)$ 
depends only on 
$[i_\ast l]_{\pi_1} \in \pi_1\left(M,\star\right)$, and this assignment defines a group homomorphism from the homotopy group to the homology group. 
Since the loops in $\bl{\Delta}^0$ have trivial homotopy type, 
${\rm C}(\bl{\Delta}^0)$ is composed by boundaries; in fact, it is easy to see that every boundary is the image under ${\rm C}$ of a contractible loop, which means that 
${\rm C}(\bl{\Delta}^0) = B_1$. 
Then, coming back to our case of interest, we see that 
the non trivial behavior comes from the existence of
loops $l \in \bl{\Delta}^A$ such that $C(l^r)=r {\rm C}(l) \in B_1$ for some integer $r$. 
The abelian group of chains $C(\bl{\Delta}^A)\subset C(\bl{\Delta})$ 
is decomposed as 
\[
{\rm C}(\bl{\Delta}^A) = (\mathbb{Z} D_1 \oplus \ldots \oplus \mathbb{Z}D_m) \oplus 
(\mathbb{Z} C_1 \oplus \ldots \oplus \mathbb{Z}C_n) 
\]
where each chain $D_\alpha$ has the property $r_\alpha D_\alpha \in B_1$ and $s D_\alpha \notin B_1$ for \\$0 < s < r_\alpha$, 
and for the chains $C_\alpha$ there is no $r \in \mathbb{N}$ such that $r C_\alpha \in B_1$.%
\footnote{
Considering ${\rm C}(\bl{\Delta}^A) $ as a module, 
the existence of this decomposition is a consequence of basic linear algebra.
} 
Given that we are interested in homomorphisms to an abelian group, the group $\h{\bl{\Delta}^A}{U(1)}$ will itself be factorized. 
Now we state the decomposition lemma. 
\begin{lma}\label{lma:decomposition}
The space of holonomy evaluations $\mathcal{A}/\mathcal{G}_{\Delta, \star}^\mathrm{Hol}\equiv \h{\bl{\Delta}}{U(1)}$ can be factorized as follows: 
\begin{description}
\item[Case 1] 
If $\bl{\Delta} = \bl{\Delta}^0 \bl{\Delta}^A$ and $\bl{\Delta}^0 \cap \bl{\Delta}^A = \{ \id \}$, then 
$$
	\mathcal{A}/\mathcal{G}_{\Delta, \star}^\mathrm{Hol} \simeq 
	\h{\bl{\Delta}^0}{U(1)} \oplus \h{\bl{\Delta}^A}{U(1)} . 
$$
\item[Case 2] 
If $\bl{\Delta} = \bl{\Delta}^0 \bl{\Delta}^A$ and 
$\bl{\Delta}^0 \cap \bl{\Delta}^A$ is a non trivial subgroup of $\bl{\Delta}$, then 
$$
	\mathcal{A}/\mathcal{G}_{\Delta, \star}^\mathrm{Hol} \simeq 
	[\h{\bl{\Delta}^0}{U(1)} \times \h{\bl{\Delta}^A}{U(1)}]_{l \in \bl{\Delta}^0 \cap \bl{\Delta}^A \Rightarrow \varphi^0(l)= \varphi^A(l)} . 
$$
\end{description}
Moreover, we can write 
\[
\h{\bl{\Delta}^A}{U(1)} \simeq 
[\oplus_{\alpha=1}^m \h{\mathbb{Z} D_\alpha}{U(1)} ] \oplus [ \oplus_{\alpha=1}^n \h{\mathbb{Z} C_\alpha}{U(1)} ]
\]
and the compatibility conditions $l \in \bl{\Delta}^0 \cap \bl{\Delta}^A \Rightarrow \varphi^0(l)= \varphi^A(l)$ are written as 
\[
\varphi^0(r_\alpha D_\alpha) = [ \varphi^{D}(D_\alpha)]^{r_\alpha} \, \, \, \, \, \forall \, \alpha \in \{ 1, \ldots ,m \}. 
\]
The proof gives the isomorphisms explicitly. 
\end{lma}
\begin{proof}
{\em Case 1}\\
Since $U(1)$ is abelian $\h{\bl{\Delta}}{U(1)}$, 
$\h{\bl{\Delta}^0}{U(1)}$ and $\h{\bl{\Delta}^A}{U(1)}$ are groups.

It is simple to verify that $\h{\bl{\Delta}^0}{U(1)} \times \{ \id \}$ 
and $\{ \id \}\times \h{\bl{\Delta}^A}{U(1)}$ 
can be embedded as abelian subgroups of 
$\h{\bl{\Delta}}{U(1)}$. 

Consider the following group homomorphism $j_A : \h{\bl{\Delta}}{U(1)} \to 
\h{\bl{\Delta}^0}{U(1)} \oplus \h{\bl{\Delta}^A}{U(1)}$,
\[
\varphi \mapsto (\varphi|_{\bl{\Delta}^0} , \varphi|_{\bl{\Delta}^A}) . 
\]
In order to prove that it is an isomorphism, simply notice that its inverse is defined by 
\[
[ j_A^{-1}(\varphi^0, \varphi^A)](l) = [ j_A^{-1}(\varphi^0, \varphi^A)] (l^0 l^A) = 
\varphi^0 (l^0)  \varphi^A (l^A) . 
\]

{\em Case 2}\\
In this case the injective group homomorphism $j_A$ defined above is not onto; its image has the property that  
\[
l \in \bl{\Delta}^0 \cap \bl{\Delta}^A \implies \varphi^0 (l) = \varphi^A (l) . 
\]
We will call this relation the {\bf compatibility} condition. 
It defines a subgroup of 
$\h{\bl{\Delta}^0}{U(1)} \times \h{\bl{\Delta}^A}{U(1)}$. 
Now we will see that 
$j_A : \h{\bl{\Delta}}{U(1)} \to \h{\bl{\Delta}^0}{U(1)} \times \h{\bl{\Delta}^A}{U(1)}|_{\rm \bf compat.}$ 
is invertible. 
Consider $(\varphi^0, \varphi^A) \in \h{\bl{\Delta}^0}{U(1)} \times \h{\bl{\Delta}^A}{U(1)}|_{\rm \bf compat.}$, 
and two different factorizations of a loop $l = l^0_1 l^A_1=l^0_2 l^A_2$. Then 
\begin{eqnarray*}
\varphi^0 (l^0_1)  \varphi^A (l^A_1) &=& 
\varphi^0 ((l^0_2 l^A_2)(l^0_1 l^A_1)^{-1}) 
\varphi^0 (l^0_1)  \varphi^A (l^A_1) \\
&=&
\varphi^0 (l^0_2) \varphi^A (l^A_2 {l^A_1}^{-1}) \varphi^0 ({l^0_1}^{-1}) \varphi^0 (l^0_1)  \varphi^A (l^A_1) \\
&=& 
\varphi^0 (l^0_2)  \varphi^A (l^A_2) . 
\end{eqnarray*}
Thus, whenever $(\varphi^0, \varphi^A) \in \h{\bl{\Delta}^0}{U(1)} \times \h{\bl{\Delta}^A}{U(1)}|_{\rm \bf compat.}$ we can define 
$j_A^{-1}(\varphi^0, \varphi^A) \in \h{\bl{\Delta}}{U(1)}$ using any factorization of the loops $l \in \bl{\Delta}$. 
It is simple to verify that $j_A^{-1}$ is in fact the inverse of 
$j_A : \h{\bl{\Delta}}{U(1)} \to \h{\bl{\Delta}^0}{U(1)} \times \h{\bl{\Delta}^A}{U(1)}|_{\rm \bf compat.}$. 

From their definition, it is clear that $j_A$ and its inverse are smooth. 

The finer factorization at the end of the lemma follows from two facts: 
(i) Since $U(1)$ is an abelian group implies that the holonomy evaluation along $l \in \bl{\Delta}$ is sensible only to the corresponding 
chain ${\rm C}(l)$. 
(ii) The decomposition 
$
{\rm C}(\bl{\Delta}^A) = (\oplus_{\alpha=1}^M \mathbb{Z} D_\alpha ) \oplus (\oplus_{\alpha=1}^N \mathbb{Z} C_\alpha) 
$
described just prior to the lemma. 
\end{proof}

\section{Holonomy evaluation and curvature evaluation} \label{sec4}

In the last section we described the connection modulo gauge $[A]\in\AGp$ at measuring scale $(|\Delta|,\phi)$ restricting the set of measurements to 
the available holonomies. 
The possible holonomy evaluations at scale $(|\Delta|,\phi)$ are stored in 
$\mathcal{A}/\mathcal{G}_{\Delta, \star}^\mathrm{Hol}$. However, as explained in the introduction, this description does not let us evaluate, even approximately, the curvature of $[A]$ on a simplicial surface contained in $\phi({\rm Sd}\,|\Delta|^{(2)})$. This motivates further considerations described below.

Let us consider the abelian group $G=U(1)$, whose one dimensional Lie algebra ${\rm Lie}(U(1))$ can be identified with $\mathbb{R}$. We consider a smooth principal $U(1)$-bundle $E$ over a compact oriented manifold $M$, with Euler class $\mathbf{e}\in H^2\left(M,\mathbb{Z}\right)$.  Recall that for each smooth connection $[A]\in\AGp$, the integral of the curvature on an oriented $2$-simplex is
$$
	\omega^A(\sigma):=\frac{1}{2\pi}\int_\sigma F^A,
$$
where, $F^A=dA$ is a smooth real valued $2$-form on $M$. This function $\omega^A:({\rm Sd}\,|\Delta|^2,\phi)\rightarrow\mathbb{R}$, can be extended 
to become a $(2, \Delta)$-cochain 
as
$$
	\omega^A\left(\sum_{k=1}^{N_{2}}r_k\sigma_k\right):=\sum_{k=1}^{N_{2}}r_k\omega^A(\sigma_k)
$$
where $\sigma_k\in ({\rm Sd}\,|\Delta|^2,\phi)$ and $r_k\in\mathbb{Z}$. 

Consider a $(2, \Delta)$-cochain $\omega$ with associated map $\omega:({\rm Sd}\,|\Delta|^2,\phi)\rightarrow\mathbb{R}$. We say that $\omega$ is \emph{closed} if for every $3$-simplex $\tau\in ({\rm Sd}\,|\Delta|^3,\phi)$ with $\partial\tau=\sigma^1+\sigma^2+\sigma^3+\sigma^4$ we have 
$$
d\omega(\tau):=\omega(\partial\tau)=
\omega(\sigma^1) + \omega(\sigma^2) + \omega(\sigma^3) + \omega(\sigma^4)
=0 . 
$$
Clearly, the curvature $(2, \Delta)$-cochains $\omega^A$ are closed, $d\omega^A=0$.

Let $S$ be an oriented compact simplicial surface $S\subset \phi({\rm Sd}\,|\Delta|^{(2)})$ such that $\partial S=\emptyset$. 
For a closed $(2, \Delta)$-cochain $\omega$, the evaluation $\omega(S)$ depends only on the homology class $[S]$. We may also define the cohomology class $[\omega]\in H^2(M,\mathbb{R}).$ 
Therefore, when we evaluate the $(2, \Delta)$-cochain $\omega^A$ on $S$, a simplicial oriented surface with empty boundary, we obtain
$$\omega^A(S)=\mathbf{e}([S])\in\mathbb{Z},$$
where $\mathbf{e}$ is the Euler class of the principal bundle.

Conversely, given a prefixed Euler class $\mathbf{e}\in H^2(M,\mathbb{Z})$, consider a closed $(2, \Delta)$-cochain $\omega$ such that $\omega(S)=\mathbf{e}([S])$ for every oriented closed surface $S$. Consider the $(2, \Delta)$-cochain $\omega\mid_\tau, d\omega\mid_\tau=0$ defined on a $d$-simplex $\tau\in({\rm Sd}\,|\Delta|^d,\phi)$; by a discrete Poincar\'e Lemma (see a more general description in \cite{DLM}), there exists a $1$-cochain $\gf_\tau$ defined on $\tau \subset M$ such that
$$d\gf_\tau=\omega\mid_\tau.$$
The set of $(1,\Delta)$-cochains $\gf_\tau$ may be thougth of as local connection discrete $1$-forms whose local curvature discrete $2$-forms are the corresponding $\omega\mid_\tau.$

So far we have motivated the definition of the set of {\em discrete $2$-forms} encoding the curvature evaluation for effective theories at a given scale $(|\Delta|,\phi)$. 
The confirmation that it provides a useful characterization of the space of curvature evaluations at a given scale is the content of Lemma \ref{lma:maps} below. 
\begin{dfn}\label{dfn:Omega}  The space of \emph{curvature evaluations} on an abelian principal bundle $(E,p,M)$ is the space of real valued $(2, \Delta)$-cochains 
$\omega\in\Omega_{\Delta,p}$ 
such that
\begin{enumerate}
	\item (Bianchi/closure) The cochain $\omega$ is closed; in other words, $d\omega(\tau)=\omega(\partial\tau)=0$ for all $3$-simplices $\tau \in ({\rm Sd}\,|\Delta|^3,\phi)$. 
	\item For each simplicial surface with $\partial S=\emptyset$ we have 
	$\omega(S)=\mathbf{e}([S])\in\mathbb{Z}$, where $\mathbf{e}$ is the Euler class of the principal bundle $(E,p,M)$.
\end{enumerate}
$\Omega_{\Delta,p}$ is an affine subspace of 
$\mathbb{R}^{N_{2}}$, where the components are in correspondence with the 
set $({\rm Sd}\,|\Delta|^2,\phi)=\{\sigma^1,\dots,\sigma^{N_2}\}$.
Its underlying vector space is the vector subspace 
$$
\Omega_{\Delta}^0\subset\mathbb{R}^{N_2}
$$ 
defined by the Bianchi/closure condition stated above, and the condition $\omega(S)=0$ for every simplicial surface with empty boundary. 
Notice that the underlying vector space does not depend on the bundle. 

We also define the space 
$$
\check{\Omega}_\Delta = \cup_p \Omega_{\Delta,p} ,
$$
which is the natural space to consider if we want to use curvature evaluation data to determine the bundle structure \cite{DZ}. It is clear that 
$\check{\Omega}_\Delta \subset\mathbb{R}^{N_2}$ is composed by 
strata homeomorphic to $\Omega_{\Delta}^0$ labeled by 
$\mathbf{e}\in H^2(M,\mathbb{Z})$. 
\end{dfn}

We conclude this section by showing how the holonomy evaluations of abelian connections for contractible loops contained in $\phi({\rm Sd}\,|\Delta|^{(1)})$ can be obtained from the curvature evaluations on simplicial surfaces $S$ contained in $\phi({\rm Sd}\,|\Delta|^{(2)})$. 
For the whole description of the holonomy evaluation along contractible and not contractive loops at a given scale, in addition to the curvature we also need the holonomy evaluation along non contractible loops. 
Roughly speaking, this description of the holonomy evaluation is given by the covering map 
$
\mathfrak{e}_\Delta : 
\Omega_{\Delta,p}\times\ \h{\bl{\Delta}^A}{U(1)}|_{\rm \bf compat.}
\rightarrow\mathcal{A}/\mathcal{G}_{\Delta, \star}^\mathrm{Hol}
$
described in detail in Lemma \ref{lma:maps} below. 
The space $\Omega_{\Delta,p}\times\ \h{\bl{\Delta}^A}{U(1)}|_{\rm \bf compat.}$ 
will be referred as {\em the extended space of curvature evaluations}. 

\begin{lma}\label{lma:maps}
1) There is an isomorphism of vector spaces 
$$
I_\Delta : \Omega_{\Delta}^0 \to \h{\bl{\Delta}^0}{\mathbb R} 
= {\rm Lie}\left[ \h{\bl{\Delta}^0}{U(1)} \right] 
$$
defined by 
$I_\Delta (\omega^0)= \frac{d}{dt} (c_{\vec{\omega}}(t))_{t=0}$ for the curve in $\h{\bl{\Delta}^0}{U(1)}$ 
determined by  
$
c_{\vec{\omega}}(t) [\gamma^0] = \exp\left(2\pi i t \vec{\omega}(S^0)\right)\in U(1)  
$, where ${\rm C} (\gamma^0) = \partial S^0$. 

2) The extended space of curvature evaluations covers 
$\mathcal{A}/\mathcal{G}_{\Delta, \star}^\mathrm{Hol}$. \\
After fixing a base point $\omega^0 \in \Omega_{\Delta,p}$, 
we can use the exponential map \\
$\exp_\Delta: \h{\bl{\Delta}^0}{\mathbb R} = 
{\rm Lie}\left[ \h{\bl{\Delta}^0}{U(1)} \right] \to \h{\bl{\Delta}^0}{U(1)}$, \\
and the decomposition \\
$j_A: 
\mathcal{A}/\mathcal{G}_{\Delta, \star}^\mathrm{Hol} \to 
	\h{\bl{\Delta}^0}{U(1)}\oplus \h{\bl{\Delta}^A}{U(1)} 
$ 
of Lemma \ref{lma:decomposition} \\
to obtain the following covering maps: 
\begin{itemize}
	\item 
$
	e_\Delta^0 = \exp_\Delta \circ I_\Delta \circ (-\omega^0)
	:\Omega_{\Delta,p} \rightarrow\h{\bl{\Delta}^0}{U(1)}
$
	\item
\begin{description}
\item[Case 1] 
If $\bl{\Delta} = \bl{\Delta}^0 \bl{\Delta}^A$ and $\bl{\Delta}^0 \cap \bl{\Delta}^A = \{ \id \}$, then 
$$
\mathfrak{e}_\Delta = j_A^{-1} \circ (e_\Delta^0 \times \id)  
:\Omega_{\Delta,p}\times\ \h{\bl{\Delta}^A}{U(1)} 
\rightarrow\mathcal{A}/\mathcal{G}_{\Delta, \star}^\mathrm{Hol} . 
$$
\item[Case 2] 
If $\bl{\Delta} = \bl{\Delta}^0 \bl{\Delta}^A$ and 
$\bl{\Delta}^0 \cap \bl{\Delta}^A$ is a non trivial subgroup of $\bl{\Delta}$, then the cover of 
$\mathcal{A}/\mathcal{G}_{\Delta, \star}^\mathrm{Hol}$ is the map 
$$
\mathfrak{e}_\Delta = j_A^{-1} \circ (e_\Delta^0 \times \exp_D \times \id) 
$$
with domain 
$$
(\Omega_{\Delta,p}
\times\ 
[\oplus_{\alpha=1}^m \h{\mathbb{Z} D_\alpha}{\mathbb R} ])|_{\rm \bf compat.}   \oplus [ \oplus_{\alpha=1}^n \h{\mathbb{Z} C_\alpha}{U(1)} ].
$$
We will use the notation $\exp_D(\tilde{\varphi}^D) = \varphi^D$ for the map \\
$\exp_D: \oplus_{\alpha=1}^m \h{\mathbb{Z} D_\alpha}{\mathbb R} \to \oplus_{\alpha=1}^m \h{\mathbb{Z} D_\alpha}{U(1)}$. \\
In this notation the compatibility conditions become 
$$
\omega(S(r_\alpha D_\alpha)) = \tilde{\varphi}^D(r_\alpha D_\alpha) 
$$
for every surface $S(r_\alpha D_\alpha)$ such that $\partial S(r_\alpha D_\alpha)= r_\alpha D_\alpha$. \\
This factorized form of the compatibility conditions has the property that the natural projection 
$$
(\Omega_{\Delta,p}
\times\ 
[\oplus_{\alpha=1}^m \h{\mathbb{Z} D_\alpha}{\mathbb R} ])|_{\rm \bf compat.}   
\to 
\Omega_{\Delta,p}
$$
is a diffeomorphism. 
\end{description}
\end{itemize}
3) The group $\Omega_\Delta^0\times U(1)^n$ acts on 
the extended space of curvature evaluations making it a homogeneous space. 
\end{lma}
\begin{proof}
First, we shall mention that 
the formula 
$
c_{\vec{\omega}}(t) [\gamma^0] = \exp\left(2\pi i t \vec{\omega}(S^0)\right)\in U(1)  
$
does not depend on the choice of surface $S^0$ with prescribed boundary 
${\rm C} (\gamma^0)$ due to the ``Bianchi/closure'' condition obeyed by 
$\vec{\omega}$. 

It is also simple to verify that $c_{\vec{\omega}}(t)$ 
defines an element of $\h{\bl{\Delta}^0}{U(1)}$ for each value of the parameter $t$.

In order to prove that $I_\Delta$ is an isomorphism we exhibit its inverse. 
An element of $ \h{\bl{\Delta}^0}{\mathbb R} = {\rm Lie}( \h{\bl{\Delta}^0}{U(1)} )$ can be viewed as 
a group homomorphism $c : \mathbb{R} \to \h{\bl{\Delta}^0}{U(1)}$; 
we define 
$$
\vec{\omega}_c (S) := \frac{1}{2\pi i} \frac{d}{dt} (c(t) [\partial S])_{t=0} . 
$$

Since, for each value of the parameter, $c(t)$ belongs to $\h{\bl{\Delta}^0}{U(1)}$, the formulas given above are compatible with the 
``Bianchi/closure conditions'' and with the conditions of $\omega(S)=0$ for every simplicial closed surface. 
Thus, we define 
$J_\Delta : \h{\bl{\Delta}^0}{\mathbb R} = 
{\rm Lie}( \h{\bl{\Delta}^0}{U(1)} ) \to \Omega_{\Delta}^0$ as 
$$J_\Delta(c) [S] := \vec{\omega}_c(S).$$ 
It is easy to verify that $J_\Delta \circ I_\Delta = \id$ and that $I_\Delta \circ J_\Delta  = \id$. 
This completes the proof of the first part of the lemma. 

It is clear that the compatibility conditions can be 
solved considering $\omega$ as data and $\tilde{\varphi}^D$ as unknown, which means that the projection 
\[
(\Omega_{\Delta,p}
\times\ 
[\oplus_{\alpha=1}^m \h{\mathbb{Z} D_\alpha}{\mathbb R} ])|_{\rm \bf compat.}   
\to 
\Omega_{\Delta,p}
\]
is onto and invertible. 

In {\em Case 1} the $\Omega_\Delta^0\times U(1)^n$ action on 
$\Omega_{\Delta,p}\times\ \h{\bl{\Delta}^A}{U(1)} $ is clear. 
In {\em Case 2} the $\Omega_\Delta^0$ action on 
$(\Omega_{\Delta,p}
\times\ 
[\oplus_{\alpha=1}^m \h{\mathbb{Z} D_\alpha}{\mathbb R} ])|_{\rm \bf compat.}
$
is the one inherited from the action on $\Omega_{\Delta,p}$. 
One can easily verify that, in both cases, the action is transitive. 
\end{proof}

%%%%%%%%%%%%%%%%%%%

\section{Coarse graining maps and the continuum limit} \label{sec3}

In previous sections we described the holonomy and curvature evaluations at a fixed scale $(|\Delta|,\phi)$. Now we consider 
a sequence of scales, $(|\Delta_i|,\phi_i),$ with $i\in\mathbb{N}$, the relation among them, and a notion of continuum limit.

\begin{dfn}\label{dfn:scales}
We consider a sequence of smooth triangulations $\phi_i:|\Delta_i|\rightarrow M$, $i\in\mathbb{N}$. We say that the sequence 
$$\mathscr{E}=\{(|\Delta_i|,\phi_i)\,\mid\,i\in\mathbb{N}\}$$ 
is a \emph{sequence of scales} for $M$, if the following conditions are satisfied:
\begin{enumerate}
\item The triangulation $(|\Delta_j|,\phi_j)$ is a rectilinear sub-triangulation of $(|\Delta_i|,\phi_i)$ for every $j>i$, i.e. for any simplex $|\tau|\subset|\Delta_i|$ the simplex $\phi_j^{-1}\circ \phi_i(|\sigma|)$ is contained in a simplex of $|\Delta_j|$, with affine inclusions $\phi_i^{-1}\circ \phi_j\mid_{|\tau|}$.
\item If we consider barycentric sub-triangulations ${\rm Sd}\,|\Delta_i|,$ of $|\Delta_i|$, the barycentric sub-triangulation ${\rm Sd}\,|\Delta_j|$ at scales $(|\Delta_j|,\phi_j), j\geq i$ are sub-triangulations of ${\rm Sd}\,|\Delta_i|$ with inclusions $\phi_j^{-1}\circ \phi_i\mid_{|K_\sigma|}$ that are affine maps on each simplex $|\tau|\subset {\rm Sd}\,|\Delta_j|$.
\item The sequence is exhaustive, i.e. for each open set $U\subset M$ there is a scale $(|\Delta_i|,\phi_i)$ such that there is a $d$-simplex $\sigma\in(|\Delta_i|^d,\phi_i)$, with $\sigma\subset U$, $d=\dim M$.
\end{enumerate}
\end{dfn}

Notice that the second condition implies that every $k$-simplex $\sigma_i$ in $({\rm Sd}\,|\Delta_i|^k,\phi_i) $ is a $k$-chain $\sum_{l=1}^m\sigma_j^l$ of $k$-simplices $\sigma_i^l\in({\rm Sd}\,|\Delta_j|^k,\phi_j)$; we say that a scale $(|\Delta_j|,\phi_j)$ is \emph{finer} than a scale $(|\Delta_i|,\phi_i)$. This relation gives a total ordering for the scales corresponding to the total ordering for $\mathbb{N}$.

\begin{paragraph}{Example.} 
{\em 
Consider a smooth triangulation $\phi:|\Delta|\rightarrow M$. If we define 
$$|\Delta_i|:={\rm Sd}^i|\Delta|$$
note that $|\Delta_{i+1}| = {\rm Sd}\,|\Delta_i|$. Let $h_{i}:{\rm Sd}^{i}|\Delta|\rightarrow {\rm Sd}^{i-1}|\Delta_i|$ denote the barycentric sub-triangulation, and define 
$\phi_i:{\rm Sd}^i|\Delta|\rightarrow M$ as $\phi_i=h_1\circ \dots\circ h_{i}$. 
This is an example of a sequence of scales, 
$$\mathscr{E}=\{(|\Delta_0|,\phi_0), (|\Delta_1|,\phi_1) , \, \ldots  \}. $$
}
\end{paragraph}

Let $\bl{i}$ be defined as in equation (\ref{eqn:p_Delta}), $\widehat{\pr_i}$ as in (\ref{eqn:pr_Delta}), and $\ba_i$ as in (\ref{eqn:ba_Delta}). 
For two scales $(|\Delta_i|,\phi_i),(|\Delta_j|,\phi_j)$, with $i\leq j$, 
define the map $\pr_{ji}=\widehat{\pr_i}\circ \ba_j(\gamma)$,
$$
	\pr_{ji}:\bl{j}\rightarrow\bl{i}.
$$
This map is an epimorphism.

\begin{dfn}
\begin{enumerate}
\item The \emph{coarse graining} of holonomy evaluations from a scale 
$(|\Delta_j|,\phi_j)$ to a coarser one $(|\Delta_i|,\phi_i)$ is 
constructed by considering a section of $\pr_{ji}$. This section is
$\enc_{ij}=\widehat{\pr_j}\circ\ba_i:\bl{i}\rightarrow\bl{{j}}$, 
and the coarse graining map is 
$$
\pi_{ji} = \enc_{ij}^\ast: 
\mathcal{A}/\mathcal{G}_{\star,\Delta_j}^\mathrm{Hol} \to 
\mathcal{A}/\mathcal{G}_{\star,\Delta_i}^\mathrm{Hol} .
$$
\item The \emph{projective limit} is the compact group
$$
	\underleftarrow{\AG}:=
	\varprojlim \mathcal{A}/\mathcal{G}_{\star,\Delta_i}^\mathrm{Hol} . 
$$
This space is closely related to the space of generalized connections of loop quantization $\overline{\AG}$. On the one hand, there is a class of generalized connections called $C$-flat which which form a dense subset of $\overline{\AG}$ 
\cite{MMZ} and certain sequences of measures defined on $\overline{\AG}$ by 
$(\hat{\pr{i}}^\ast)_\ast$ may converge to interesting measures in 
$\overline{\AG}$ \cite{MOWZ}. 
On the other hand, the maps $\bar{i}$ bring generalized connections to any scale in a compatible manner, and thus to $\underleftarrow{\AG}$. 
\item Consider the \emph{coarse graining} of curvature evaluations by the 
surjective affine maps 
$$
\ti{\pi_{ji}}= \ti{\enc}_{ij}^\ast:\Omega_{j,p}\rightarrow\Omega_{i,p} , 
$$ 
which are the dual to the affine inclusions $\ti{\enc}_{ij}$ acting on 
$2$-chains at scale $(|\Delta_i|,\phi_i)$ by 
$$
	\ti{\enc}_{ij}(\sigma)=\sigma^1+\dots+\sigma^r,
$$
where for each $2$-simplex $\sigma\in ({\rm Sd}\,|\Delta_i|^2,\phi_i)$ we associate the $2$-chain $\sigma^1+\dots+\sigma^r,\,\sigma^k\in ({\rm Sd}\,|\Delta_j|^2,\phi_j),$ $\sigma=\cup_{k=1}^r\sigma^k$.

\item We describe the configurations for the curvature evaluations in the continuum limit as the projective limit of affine spaces
$$
	\underleftarrow{\Omega}_p:=\varprojlim\Omega_{i,p}.
$$
It is also an affine space with $\underleftarrow{\Omega}^0:=\varprojlim\Omega_i^0$
as underlying vector space.

\item We also define 
$$
	\underleftarrow{\check{\Omega}}:=\varprojlim\check{\Omega}_i.
$$
\end{enumerate}
\end{dfn}

In order to extend the results of Lemma \ref{lma:decomposition} to the continuum limit, we have to study the relation between coarse graining and the factorization of loops as contractible and non-contractible, and also consider the relation between the coarse graining map and the more subtle factorization given just prior to Lemma \ref{lma:decomposition}. 

We define $\bl{i}^0\subset \bl{i}$ as in Definition \ref{L^0}, and 
$\bl{M}^0 \subset \bl{M}$ as the subgroup of contractible loops. 
We will assume that at the coarsest scale a choice for $\bl{0}^A$ has been made; for the finer scales we define 
$$
\bl{i}^A:= \enc_{0i}(\bl{0}^A) \subset \bl{i} \, , \, 
\bl{M}^A:= \ba_i(\bl{0}^A) \subset \bl{M} \, .
$$ 

Notice that, form our definitions, it is clear that 
$$
\pr_{ji}(\bl{j}^0) = \bl{i}^0 \, , \, 
\enc_{ij}(\bl{i}^0) \subset \bl{j}^0 \, , \, 
\pr_{ji}(\bl{j}^A) = \bl{i}^A \, , \, 
\enc_{ij}(\bl{i}^A) = \bl{j}^A \, . 
$$
Also notice that, since 
$
\pi_1\left(M,\star\right) \simeq \bl{i} / \bl{i}^0
$, 
and $\bl{i}^A$ contains representatives of all $\pi_1(M, \star)$-classes, 
we know that $\bl{i} = \bl{i}^0 \bl{i}^A$. 

We can apply the argument used at scale $i= 0$ in an independent manner at each scale to obtain as a result 
that each abelian group of chains $C(\bl{i}^A)\subset C(\bl{i})$ 
is decomposed as 
\[
{\rm C}(\bl{i}^A) = (\mathbb{Z} D_{i, 1} \oplus \ldots \oplus \mathbb{Z}D_{i, m}) \oplus 
(\mathbb{Z} C_{i, 1} \oplus \ldots \oplus \mathbb{Z}C_{i, n}) 
\]
where each chain $D_{i, \alpha}$ has the property $r_{i, \alpha} D_{i, \alpha} \in B_1$ and $s D_{i, \alpha} \notin B_1$ for \\$0 < s < r_{i, \alpha}$, 
and for the chains $C_{i, \alpha}$ there is no $r \in \mathbb{N}$ such that 
$r C_{i, \alpha} \in B_1$.

The good news are that all these decompositions are not independent; instead, they are essentially the same. More precisely, in the decompositions we can use 
\[
D_{j, \alpha} = \enc_{ij} (D_{i, \alpha}) \, 
\mbox{ with } r_{j, \alpha}= r_{i, \alpha} \,
\mbox{ and } \, C_{j, \alpha} = \enc_{ij} (C_{i, \alpha}) \, .
\]

These properties imply that Lemma \ref{lma:decomposition} is compatible with the coarse graining operation and with the continuum limit. 
Thus, the reader should review Lemma \ref{lma:decomposition} changing the subscript $\Delta$ for $i$ --the subscript that we use to denote a given scale in the sequence; in addition, the lemma also holds in the continuum where $\Delta$ has to be replaced by $M$, and the space being covered is $\underleftarrow{\AG}$ 
which, as we just mentioned, is closely related to the space of generalized connections of loop quantization $\overline{\AG}$. 

An important element which appears in the proof of Lemma \ref{lma:decomposition} is the group homomorphism 
$j_A : \h{\bl{i}}{U(1)} \to 
\h{\bl{i}^0}{U(1)} \oplus \h{\bl{i}^A}{U(1)}$, given by 
$
\varphi \mapsto (\varphi|_{\bl{i}^0} , \varphi|_{\bl{i}^A}) 
$. 
We have one such group homomorphism at each scale; 
they commute with coarse graining, which implies that they induce a similar group homomorphism in the continuum (projective) limit. 

It is worth mentioning that the division into cases of Lemma \ref{lma:decomposition} is maintained by coarse graining. That is, if at any scale $\bl{i}$ satisfies the conditions $\bl{i} = \bl{i}^0 \bl{i}^A$ and $\bl{i}^0 \cap \bl{i}^A = \{ \id \}$ 
(Case 1), then the same condition is satisfied at all the scales; 
we should write 
$\bl{M} = \bl{M}^0 \bl{M}^A$ and $\bl{M}^0 \cap \bl{M}^A = \{ \id \}$. 

In Case 2, 
$\bl{M} = \bl{M}^0 \bl{M}^A$ and $\bl{M}^0 \cap \bl{M}^A$ being a non trivial subgroup of $\bl{M}$, it becomes essential to talk about the compatibility conditions, 
$l \in \bl{M}^0 \cap \bl{M}^A \Rightarrow \varphi^0(l)= \varphi^A(l)$. 
Here we have the important property that if the connection at scale $j$ is brought to coarser scales by the coarse graining map and the connection at scale $j$ satisfies the compatibility condition, then the coarser coarser connection will also satisfy the condition. Moreover, the refined statement of the compatibility conditions in terms of the factorization of $\bl{j}$, 
\[
\varphi^0(r_\alpha D_\alpha) = [ \varphi^{D}(D_\alpha)]^{r_\alpha} \, \, \, \, \, \forall \, \alpha \in \{ 1, \ldots ,M \}. 
\] 
is also preserved by the coarse graining map; this is why we do not bother to write the scale as a subscript. 
%

%%%
Now we show that the results of Lemma \ref{lma:maps} 
extend to the continuum limit; we state it as a theorem because it is a central result for the rest of the article. 
\begin{tma}\label{t1}
\begin{enumerate}
	\item The map
	$$
		e^0 = \varprojlim e_i^0:
		\underleftarrow{\Omega}_p\rightarrow\varprojlim \h{\bl{i}^0}{U(1)}
	$$
	is a covering map.
	\item The extended space of curvature evaluations covers 
	$\underleftarrow{\AG}$. 
%%%
\begin{description}
\item[Case 1] 
If $\bl{M} = \bl{M}^0 \bl{M}^A$ and $\bl{M}^0 \cap \bl{M}^A = \{ \id \}$, the covering map is 
$$
\mathfrak{e}= j_A^{-1} \circ (e^0 \times \id)  : 
\underleftarrow{\Omega}_p \times \h{\bl{M}^A}{U(1)} 
\rightarrow 
\underleftarrow{\AG}  . 
$$
\item[Case 2] 
If $\bl{M} = \bl{M}^0 \bl{M}^A$ and 
$\bl{M}^0 \cap \bl{M}^A$ is a non trivial subgroup of $\bl{M}$, then the cover of 
$\underleftarrow{\AG}$ is the map 
$$
\mathfrak{e} = j_A^{-1} \circ (e^0 \times \exp_D \times \id) 
$$
with domain 
$$
(\underleftarrow{\Omega}_p
\times 
[\oplus_{\alpha=1}^m \h{\mathbb{Z} D_\alpha}{\mathbb R} ])|_{\rm \bf compat.}   \oplus [ \oplus_{\alpha=1}^n \h{\mathbb{Z} C_\alpha}{U(1)} ].
$$
We will use the notation $\exp_D(\tilde{\varphi}^D) = \varphi^D$ for the map \\
$\exp_D: \oplus_{\alpha=1}^m \h{\mathbb{Z} D_\alpha}{\mathbb R} \to \oplus_{\alpha=1}^m \h{\mathbb{Z} D_\alpha}{U(1)}$. \\
In this notation the compatibility conditions become 
$$
\omega(S(r_\alpha D_\alpha)) = \tilde{\varphi}^D(r_\alpha D_\alpha) 
$$
for every surface $S(r_\alpha D_\alpha)$ such that $\partial S(r_\alpha D_\alpha)= r_\alpha D_\alpha$. \\
This factorized form of the compatibility conditions has the property that the natural projection 
$$
(\underleftarrow{\Omega}_p
\times 
[\oplus_{\alpha=1}^m \h{\mathbb{Z} D_\alpha}{\mathbb R} ])|_{\rm \bf compat.}     
\to 
\underleftarrow{\Omega}_p
$$
is a diffeomorphism. 
\end{description}
%%%
\end{enumerate}
The group $\underleftarrow{\Omega}^0\times U(1)^n$ acts on 
the extended space of curvature evaluations 
making it a homogeneous space. 
\end{tma}
\begin{proof}
The affine maps $\ti{\pi_{ji}}:\Omega_{j,p}\rightarrow\Omega_{i,p}$, are such that the following diagrams commute: 
\begin{equation}
\label{eqn:diagram1}
\xymatrix{
		\Omega_{j,p}
			\ar[r]^{\ti{\pi_{ji}}}\ar[d]^{e_{j}^0}
	&	
		\Omega_{i,p}
			\ar[d]^{e_i^0}
	\\
		\h{\bl{{j}}^0}{U(1)}
			\ar[r]^{\pi_{ji}}
	&
		\h{\bl{i}^0}{U(1)}
}.
\end{equation}
The commutativity of the diagram (\ref{eqn:diagram1}) 
allows us to define 
$e^0$ and $\mathfrak{e}$ as in Lemma \ref{lma:maps}. Since the maps $e_i^0$ and the maps 
$\mathfrak{e}_i$
are covering maps, so are the maps $e^0$ and $\mathfrak{e}$. 

The rest of the proof continues in complete analogy with the proof of Lemma \ref{lma:maps}. 

It is clear that the compatibility conditions can be 
solved considering $\omega$ as data and $\tilde{\varphi}^D$ as unknown, which means that the projection 
\[
(\underleftarrow{\Omega}_p
\times 
[\oplus_{\alpha=1}^m \h{\mathbb{Z} D_\alpha}{\mathbb R} ])|_{\rm \bf compat.}     
\to 
\underleftarrow{\Omega}_p
\]
is onto and invertible. 

% From proof of lemma 3
In {\em Case 1} the $\underleftarrow{\Omega}^0\times U(1)^n$ action on 
$\underleftarrow{\Omega}_p \times\ \h{\bl{\Delta}^A}{U(1)} $ is clear. 
In {\em Case 2} the $\underleftarrow{\Omega}^0$ action on 
$
(\underleftarrow{\Omega}_p
\times 
[\oplus_{\alpha=1}^m \h{\mathbb{Z} D_\alpha}{\mathbb R} ])|_{\rm \bf compat.}
$
is the one inherited from the action on $\underleftarrow{\Omega}_p$. 
One can easily verify that, in both cases, the action is transitive. 

It is also instructive to look at the relation between the group action at different scales with the coarse graining maps. The action is compatible with coarse graining if one defines the natural coarsening of the symmetry group from the finer to the coarser scale. 
\end{proof}
Now we comment on the meaning of the results of this section assuming Case 1, $\bl{M} = \bl{M}^0 \bl{M}^A$ and $\bl{M}^0 \cap \bl{M}^A = \{ \id \}$. Just before closing the section we will explain how to adapt the statements to the general case.

The space $\underleftarrow{\Omega}_p\times \h{\bl{M}^A}{U(1)}$ 
parametrizes configurations in 
the continuum limit of a collection of effective theories. If we take a smooth connection modulo gauge $[A]\in\AGp$, the sequence 
$\omega_i^A\in\Omega_{i,p}$ given by
$$
	\omega_i^A(\sigma):=\frac{1}{2\pi}\int_{\sigma} F^A,\qquad
	\sigma\in ({\rm Sd}\,|\Delta_i|^2,\phi_i) 
$$
satisfies $\ti{\pi}_{ji}(\omega_j^A)=\omega_i^A$. Hence, it defines an evaluation of the curvature in the continuum limit $\omega^A\in\underleftarrow{\Omega}_p$. Two classes of connections $[A],[A']$ having the same curvature form $F^A=F^{A'}$ have the same image $\omega^A=\omega^{A'}$. In addition, given two classes $[A],[A']$ having different curvature $2$-forms, $F^A\neq F^{A'}$, we have a sufficiently fine scale $i$ such that $\omega_i^A(\sigma)\neq\omega_i^{A'}(\sigma)$ for certain $\sigma\in ({\rm Sd}\,|\Delta_i|,\phi_i)$. These considerations motivate the following definition. 

\begin{dfn}\label{dfn:curv} For every connection modulo gauge $[A]\in\AGp$, we define $\mathrm{Curv}([A])=\omega^A$. Thus, we define a \emph{curvature evaluation map}
$$
	\mathrm{Curv}: \AGp \rightarrow \underleftarrow{\check{\Omega}} .
$$
This map is not injective since two connections may have the same curvature form and at the same time define different holonomies for non contractible loops. Thus, we define an \emph{extended curvature map} as $\ti{\mathrm{Curv}}([A])=(\omega^A, \mathrm{Hol}_A|_{\bl{M}^A})$,
$$
	\ti{\mathrm{Curv}}:\AGp\rightarrow\underleftarrow{\check{\Omega}}
	\times \h{\bl{M}^A}{U(1)}.
$$
Since every connection is completely described by the curvature form and the holonomy on loops in $\bl{M}^A$, 
the extended curvature map is actually injective. 
\end{dfn}
It is clear that $\mathrm{Curv} (\AGp) \subset  \underleftarrow{\Omega}_p$, and that 
$\ti{\mathrm{Curv}} (\AGp)  \subset  
\underleftarrow{\Omega}_p 
\times \h{\bl{M}^A}{U(1)}$. 

The homomorphisms
$$
\ba_i:\bl{i}\rightarrow\bl{M}
$$
are compatible, $\ba_i=\ba_j\circ \enc_{ij}$. The collection of maps $\h{\ba_i}{G}$ then defines a map $\beta:\overline{\AG}\rightarrow\underleftarrow{\AG}$.

When Case 1 holds, we can 
summarize the results contained in this section in the following commutative diagram involving the covering map $\mathfrak{e}$: 
\begin{equation*}
	\xymatrix{
		\AGp
			\ar[r]^{\mathrm{Hol}}\ar[d]^{\ti{\mathrm{Curv}}}
	&
		\overline{\AG}
			\ar[d]^{\beta}
	\\
		\underleftarrow{\Omega}_p\times \h{\bl{M}^A}{U(1)}
			\ar[r]^{\; \; \; \; \; \; \; \; \; \; \; \; \; \; \; \; \; \mathfrak{e}}
	&
		\underleftarrow{\AG}
	}.
\end{equation*}

If we are in Case 2 --$\bl{M} = \bl{M}^0 \bl{M}^A$ and 
$\bl{M}^0 \cap \bl{M}^A$ is a non trivial subgroup of $\bl{M}$-- the essence of what we wrote in the last paragraphs is maintained. We just have to refine the expressions. Here are the two necessary replacements: \\
1) instead of $\underleftarrow{\Omega}_p\times \h{\bl{M}^A}{U(1)}$, the relevant space is \\
$
(\underleftarrow{\Omega}_p
\times 
[\oplus_{\alpha=1}^m \h{\mathbb{Z} D_\alpha}{\mathbb R} ])|_{\rm \bf compat.}   \oplus [ \oplus_{\alpha=1}^n \h{\mathbb{Z} C_\alpha}{U(1)} ] 
$; \\
2) instead of defining the extended curvature map  
$\ti{\mathrm{Curv}}$ as done above, 
$$
\ti{\mathrm{Curv}}:\AGp\rightarrow
(\underleftarrow{\Omega}_p
\times 
[\oplus_{\alpha=1}^m \h{\mathbb{Z} D_\alpha}{\mathbb R} ])|_{\rm \bf compat.}   \oplus [ \oplus_{\alpha=1}^n \h{\mathbb{Z} C_\alpha}{U(1)} ] 
$$
is defined by 
$\ti{\mathrm{Curv}}([A])=(\omega^A, 
\exp_D(
\mathrm{Hol}_A|_D), \mathrm{Hol}_A|_C)$, 
where $\mathrm{Hol}_A|_D$ and $\mathrm{Hol}_A|_C$ mean the holonomy evaluation map restricted to the $D$ and $C$ factors of $C(\bl{M}^A)$ respectively.

%%%%%%%%%%%%%%%%%%%%%%%%%%%%%%%%%%%%%%%%%%%%%%%

\section{Gaussian kinematical measures} \label{sec5}
The purpose of this section is to give an example of a family of gaussian probability measures on the spaces that we described in the previous sections. 
When we work at a given scale $i$, Lemma \ref{lma:maps} completely characterizes the homogeneous measure of interest. First, the space where we store extended curvature evaluations at scale $i$ is finite dimensional and homogeneous which let us talk about homogeneous measures. Second, the covering map brings the measure of the covering space to the compact space 
$\mathcal{A}/\mathcal{G}_{\Delta, \star}^\mathrm{Hol}$ which is a compact homogeneous space with a unique normalized homogeneous measure. 
A gaussian measure at that scale can be constructed just by multiplying by a gaussian weight. 
In this article we model weights inspired in actions that are quadratic in the curvature. The most delicate point in constructing a measure in the continuum, 
constructed via a projective limit, 
is the fulfillment of the cylindrical consistency conditions between measures at different scales. 
Theorem \ref{t1} is the key to this crucial issue, since it 
assures us that the results of Lemma \ref{lma:maps} are compatible with the coarse graining maps and the continuum limit. 
We will work out the details assuming we are dealing with {\em Case 1} described above, 
on the space 
$$
\underleftarrow{\Omega}_p \times \h{\bl{M}^A}{U(1)} \simeq
\underleftarrow{\Omega}_p \times U(1)^n ,
$$
because the clean separation of the connection degrees of freedom into local and global will let us focus on the key issues fast. At the end of the section we will comment on how the same ideas let us treat {\em Case 2} without any new conceptual difficulty.

We will start describing $\underleftarrow{\Omega}_p$ as a measurable space. 
We are given a sequence of measures, one measure $\rho_i$ on each of the spaces 
$\Omega_{i,p}$, that is compatible 
with coarse graining in the sense that $(\ti{\pi}_{ji})_\ast\rho_j=\rho_i$. In this situation, the sequence defines a \emph{cylindrical measure,} $\rho$ on $\underleftarrow{\Omega}_p$ as explained below. 
Consider subsets of the form $\ti{\pi}_i^{-1}(U)\subset\underleftarrow{\Omega}_p$ with $U\subset \Omega_{i,p}$ a borelian set. Unions and intersections of subsets of this type can be used to generate a $\sigma$-algebra. A cylindrical measure $\rho$ on $\underleftarrow{\Omega}_p$ is a finite, additive, positive valued function on the mentioned algebra. It is characterized by the measure that it assigns to the subsets of the described type (see for example \cite{Xi}). 

Let us briefly give two remarks. The first is that 
we do not work on $\underleftarrow{\check{\Omega}}$ because we can give a measure to each of the strata $\underleftarrow{\Omega}_p$ separately. If we are interested in considering $\underleftarrow{\check{\Omega}}$ as the space of histories, we will have to supplement the collection of measures on each stratum with a measure on the set of labels, classes of bundles $(E, p, M)$ up to equivalence, of the strata. 
The second remark is that the measures that we construct in this section are kinematical measures. One simple objective achieved by their construction is to exhibit well behaved, well understood measures in the spaces of interest. 
It may happen that these measures turn out to be useful in the construction of physically relevant measures. 
In the closing section we will comment on the construction of measures of physical interest through a renormalization procedure.

%%% This info is now part of defn 4
Let us start fixing notation to describe the spaces of curvature evaluations and their coarse graining maps. 
Consider the set of functions $w:({\rm Sd}\,|\Delta_i|^2,\phi_i)\rightarrow\mathbb{R}$ as points in $\mathbb{R}^{N_{2,i}}$, where the components are in correspondence with the set $({\rm Sd}\,|\Delta|^2,\phi_i)=\{\sigma^1,\dots,\sigma^{N_{2,i}}\}$. Choose a point $\omega^0\in\underleftarrow{\Omega}_p$ and let $\omega_i^0:=\ti{\pi}_i(\omega^0)\in\Omega_{i,p}$. In this way, $\Omega_{i,p}$ can be seen as the affine subspace 
of $\mathbb{R}^{N_{2,i}}$ whose points satisfy the conditions of Definition \ref{dfn:Omega}. 
We can identify the affine space $\Omega_{i,p}$ with a finite dimensional vector space $\Omega_i^0\subset\mathbb{R}^{N_{2,i}}$. 
The correspondence between $\Omega_i^0$ and the affine space $\Omega_{i,p}$ is given by the translation from $\omega_i^0$, i.e. for every $\vec{\omega_i}\in\Omega_i^0$, we have
\begin{equation*}
	\omega_i = \omega_i^0 + \vec{\omega_i}
	\in\Omega_{i,p} .
\end{equation*}
%%%
We have linear maps $\ti{\pi}_{ji}:\Omega_j^0\rightarrow\Omega_i^0$ induced by the affine coarse graining maps $\ti{\pi}_{ji}:\Omega_{j,p}\rightarrow\Omega_{i,p}$, and we have the 
corresponding projective limit of topological vector spaces 
$\underleftarrow{\Omega}^0 = \varprojlim \Omega_i^0$ which is 
the underlying vector space of $\underleftarrow{\Omega}_p$.

Consider the dual spaces $\Omega_i'=\h{\Omega_i^0}{\mathbb{R}}$, and the corresponding dual maps $\pi_{ij}'=\h{\ti{\pi}_{ji}}{\mathbb{R}}:\Omega_i'\rightarrow\Omega_j'$, $i<j$. We can construct their injective limit 
$$
	\underrightarrow{\Omega'}:=\varinjlim\Omega_i'.
$$
We can see $\underrightarrow{\Omega'}$ as the space of $\omega^0$-homogeneous affine functions on $\underleftarrow{\Omega}_p$. It is spanned by function of the type 
$$
(\pi_i' \check{f}) [ \omega] = \check{f} ( \pi_i \, \omega - \pi_i \, \omega^0 )
$$
induced by any $\check{f} \in \Omega_i'$ for any $i$. Notice that a simple corollary of definitions says that the elements of $\underrightarrow{\Omega'}$ separate points in 
$\underleftarrow{\Omega}_p$. 
The elements of $\Omega_i'$ can be parametrized as 
$\check{f} = \sum_k r_k e^k$ with each basis vector $e^k$ acting on 
$\mathbb{R}^{N_{2,i}}$ by the extraction of the $k$-th component of the vectors. 
Furthermore, 
the coefficients $\{ r_k \}$ used in the parametrization 
satisfy the set of (redundant) conditions of the following two types: 
(i) $\sum_{\sigma^k \subset \partial \tau} r_k = 0$ 
for all $3$-simplices $\tau \in ({\rm Sd}\,|\Delta|^3,\phi)$, 
(ii) $\sum_{\sigma^k \subset S} r_k = 0$ 
for all simplicial surfaces without boundary, $\partial S = \emptyset$. 
%

%With this parametrization of elements of $\Omega_i'$, it is easy to understand the following important property: 

We can use the map 
$\mathrm{Curv}\circ \pi_i'$ to bring smooth connections to $\Omega_{i,p}$, thus, we can act with smooth connections on elements of $\Omega_i'$. 
An important property is that smooth connections separate points in $\Omega_{i,p}$, and since this is true for every $i$, smooth connections separate points in 
$\underrightarrow{\Omega'}$. In the following lemma we state this property in a manner that will be most useful later on. 
\begin{lma}\label{ImageCurvIsDense}
$\mathrm{Curv}(\AGp)$ 
is a dense subset of $\underleftarrow{\Omega}_p$. 
\end{lma}
%
%%%%%%%%%%

Now we give an example of a gaussian measure $\rho$ on $\underleftarrow{\Omega}_p$ centered on $\omega^0$ constructed using a riemannian metric on $M$. The measure 
will be defined by a sequence of measures compatible with coarse graining as explained above. 
Let $\rho_i$ be the measure on $\Omega_{i,p}$ whose characteristic function is given by 
\begin{equation}
\label{FTi}
\tilde{\rho}_i(f) = \int_{\Omega_{i,p}} \exp(i f (\omega)) d\rho_i(\omega) = 
\exp(\frac{-1}{2} Q_i(\check{f}, \check{f})), 
\end{equation}
where $f: \Omega_{i,p} \to \mathbb{R}$ 
is defined by $f(\omega_i) = \check{f}(\omega_i - \omega_i^0)$ for a linear function $\check{f} \in \Omega_i'$, and $Q_i$ is the positive definite bilinear form on $\Omega_i'$ 
defined as follows: 
$$
Q_i(\sum_k r_k e^k , \sum_{k'} s_{k}' e^{k'} ) = \sum_k r_k s_k a^i_k , 
$$
where $a^i_k$ is the area of $\sigma^k_i$ according to the given metric, and 
$\check{f} = \sum_k r_k e^k \in \Omega_i'$. 
	
%%%%%%%
%%%%%%%
At this point, for each scale we have a bilinear form $Q_i$ on $\Omega_i'$ and a gaussian measure $\rho_i$ on $\Omega_{i,p}$. These structures are compatible with the coarse graining map 
$\ti{\pi_{ji}} : \Omega_{j,p} \to \Omega_{i,p}$ in two ways. First, the dual coarse graining map 
is an isometry, $\pi_{ij}^{'\ast} Q_j = Q_i$. This can be verified directly from the definition. 
Second, the measures are compatible with coarse graining, \\
$\ti{\pi_{ji}}_\ast \rho_j = \rho_i$, which is easily verified from (\ref{FTi}). 
Compatibility implies that on 
$\underrightarrow{\Omega'}$ we have defined a bilinear form $Q$, and on 
$\underleftarrow{\Omega}_p$ we have defined a (projective) gaussian measure $\rho$. 

Recall that the Fourier transform of a measurable real function 
$f:\underleftarrow{\Omega}_p\rightarrow \mathbb{R}$ in the measure space $\underleftarrow{\Omega}_p$ is defined as the mean of $e^{ i f}$
$$
	\int_{\underleftarrow{\Omega}_p}\exp\left( i f(\omega)\right)d\rho.
$$
An affine function $f:\underleftarrow{\Omega}_p \rightarrow \mathbb{R}$ can be written as 
$f(\omega) = \check{f} (\omega - \omega_0)$ for a linear function 
$\check{f}\in\underrightarrow{\Omega'}$. 
From equation (\ref{FTi}), we can verify that the Fourier transform of $f$ is 
$\exp\left(-Q(\check{f},\check{f})/2\right)$. Below we state this important fact as a theorem. 
\begin{tma}\label{tma:Q}
A riemannian metric on $M$ induces a gaussian measure $\rho$ on $\underleftarrow{\Omega}_p$ centered in $\omega^0\in\underleftarrow{\Omega}_p$. The Fourier transform of any affine function 
$f:\underleftarrow{\Omega}_p \rightarrow \mathbb{R}$ 
of the type described above 
is $\exp(-Q(\check{f},\check{f})/2)$, where $Q$ is the positive definite bilinear form induced on $\underrightarrow{\Omega'}$ by the collection of compatible bilinear forms $Q_i$ constructed using the riemannian metric. 
\end{tma}
In {\em Case 1} the space of interest is $\underleftarrow{\Omega}_p \times \h{\bl{M}^A}{U(1)}$. 
At each scale, the physically motivated measures 
in $\Omega_{i,p} \times \h{\bl{i}^A}{U(1)}$
come from a uniform measure 
multiplied by a weighting factor calculated by some local formula. 
Thus, in the continuum 
the natural measures to consider 
are products of the gaussian measures constructed for $\underleftarrow{\Omega}_p$ and the Haar measure on $\h{\bl{M}^A}{U(1)} \simeq U(1)^n$. 
Notice that the soundness of this argument, depends crucially on 
the compatibility of the 
$\Omega_i^0 \times U(1)^n$ action at different scales 
with the the coarse graining maps, as stablished in the proof of Theorem \ref{t1}. 
Now we comment on {\em Case 2} described in Lemmas \ref{lma:decomposition} and \ref{lma:maps}, where 
$\bl{\Delta} = \bl{\Delta}^0 \bl{\Delta}^A$ and 
$\bl{\Delta}^0 \cap \bl{\Delta}^A$ is a non trivial subgroup of $\bl{\Delta}$. 
The space of interest is 
$
(\underleftarrow{\Omega}_p
\times 
[\oplus_{\alpha=1}^m \h{\mathbb{Z} D_\alpha}{\mathbb R} ])|_{\rm \bf compat.}   \oplus [ \oplus_{\alpha=1}^n \h{\mathbb{Z} C_\alpha}{U(1)} ]
$. 
The construction follows the logic of fixing a uniform measure 
at each scale (which can be done thanks to the results of Lemma \ref{lma:maps}), multiply it by a weighting factor calculated by some local formula which may be, for example, a gaussian weight, and considering the induced measure in the continuum limit space (when the cylindrical consistency conditions are satisfied). Thus, for the $C$-factor we will consider the Haar measure, and for the other factor we will consider a gaussian measure in 
$
(\underleftarrow{\Omega}_p
\times 
[\oplus_{\alpha=1}^m \h{\mathbb{Z} D_\alpha}{\mathbb R} ])|_{\rm \bf compat.}
$ 
using the technique we just developed for $\underleftarrow{\Omega}_p$. 
This is possible because, as shown in Theorem \ref{t1}, we can think of the compatibility conditions as a map from 
$\underleftarrow{\Omega}_p$ to 
$
(\underleftarrow{\Omega}_p
\times 
[\oplus_{\alpha=1}^m \h{\mathbb{Z} D_\alpha}{\mathbb R} ])|_{\rm \bf compat.}
$ 
and the measure can be pushed forward with this map. 
We remark that the subtle point was to find the correct covering space in {\em Case 2} and factorize it as we did in Lemmas \ref{lma:decomposition} and \ref{lma:maps}.

Even though the defined cylinder measure is of gaussian form, we do not know much about the resulting measure space. 
In the case of a two dimensional base space we do have a result showing that, in essence, $( \underleftarrow{\Omega}_p , \rho )$ is a well known space \cite{Di}. 
\begin{dfn}\label{support}
Let $\Omega_\rho\subset\underleftarrow{\Omega}_p$ be the affine subspace that is the support of the gaussian measure $\rho$. 
Denote its corresponding subspace of the underlying vector space 
by $\Omega_\rho^0 = \Omega_\rho - \omega^0 \subset 
\underleftarrow{\Omega}^0$. 
\end{dfn}
%%%
%%%
%%%
\begin{tma}\label{2d}
In the case $\dim M=2$ 
$$
\Omega_\rho^0 \subset \Omega_{bv}^0 \subset \underleftarrow{\Omega}^0 , 
$$
where $\Omega_{bv}^0$ is a reflexive separable Banach space. 
\end{tma}
%

%%%%%%%%%%%%%%%%%%%%%%%%%%%%%%%%%%%%%%%%%

\section{A cubical sequence of scales and convergence}\label{convergence}

The objective of this section is to prepare the framework to prove a convergence result for the evaluation of certain sequences of functions of the connection, generated 
by regularization, as the scale becomes finer. 
This convergence result will be central for the construction of the continuum limit presented in the next section. 
To reach our objective we will need to adjust our definition of sequence of scales (Definition \ref{dfn:scales}).
A variation of the simplicial approximation theorem is natural in our framework and leads us to the definition of $i$th representatives of curves and surfaces with boundary. That is, at a given scale $i$, to each piecewise linear curve $c \subset M$ (or surface $S \subset M$) we assign a representative $c_i \subset ({\rm Sd}\, P_i)^{(1)}$ 
(respectively $S_i \subset ({\rm Sd}\, P_i)^{(2)}$) 
where ${\rm Sd}\, P_i$ denotes a triangulation of $M$ constructed by the baricentric subdivision of a cellular decomposition $P_i$ of $M$ with cubical cells. 
Cubical cells are employed because they will let us 
generate the whole sequence of cellular decompositions using 
iterated cartesian bipartition. The relevance in the change of refinement process 
is that 
the shapes of the cells do not become degenerate as the scale is refined, in contrast to what happens if a sequence of cellular decompositions is generated by the iterated application of baricentric subdivision. 

\begin{dfn}[Cubical sequence of scales]
Consider a smooth triangulation of the base manifold $\phi:{\rm Sd}\,|\Delta|\rightarrow M$, that we will denote by $(|\Delta|, \phi)$. 
Now we will define a sequence of scales. 
The starting point of the sequence $(P_0, \phi)$, or $P_0$ to simplify the notation, is a subdivision of $(|\Delta|, \phi)$ with cubical cells. Each $n$ dimensional cell $\tau(\nu, v) \in P_0$ arises by subdivision of an $n$ dimensional simplex $\nu \in(|\Delta|, \phi)$ into $n$ dimensional cubes; there is one subcell per vertex $v\in \nu$. Thus, the cell $\tau(\nu, v)$ can be constructed as the union of the $n$ $n$ dimensional simplices of $({\rm Sd}|\Delta|, \phi)$ which lie inside the simplex $\nu \in(|\Delta|, \phi)$  and contain vertex $v\in \nu$.  

The finer cellular decompositions are generated by iterated cartesian bipartition of each of the cells of $P_0$. 
Each $n$ dimensional cell $\tau(\nu, v) \in P_{i+1}$ arises by subdivision of an $n$ dimensional cell $\nu \in P_i$ into $n$ dimensional cubes with sides of half the length; there is one subcell per vertex $v\in \nu$. 
Notice that the cell $\tau(\nu, v)$ can be constructed as the union of the $n$ $n$ dimensional simplices of ${\rm Sd}\,P_i$ which lie inside the cell $\nu \in P_i$  and contain vertex $v\in \nu$. Here we define the baricentric subdivision of $P_i$ using the affine structure of the cells in $P_i$ which is inherited from the cells of $(|\Delta|, \phi)$. 

We have defined a sequence 
$$
\mathscr{E}=\{P_i  \,\mid\,i\in\mathbb{N}\}
$$
of cellular decompositions of $M$ which we will call scales. 
Notice that if $j \geq i$, then $P_j$ refines $P_i$ and ${\rm Sd}P_j$ refines ${\rm Sd}P_i$. Also, the sequence is exhaustive in the sense that 
for each open set $U\subset M$ there is a scale $P_i$ 
and a $d=\dim M$ dimensional simplex $\sigma\in ({\rm Sd}\,P_i)^d$ with 
$\sigma\subset U$. 
\end{dfn}

In the case in which the base manifold is the torus, $M \sim T^d$, the initial cellular decomposition $P_0$ can be the usual presentation of the torus with a single $d$ cell. 

We are interested in functions of the connection associated with curves and surfaces in $M$. At each scale a regularization of the function will be constructed from a ``regularization'' of the curve or surface which is the best approximation to it at  the given scale. Below we make this concept concrete. 

\begin{dfn}\label{iRep}
The $i$th representative, $i$Rep for short, of a piecewise smooth curve $c \subset M$ (if it exists) 
is the only simplicial curve $c_i \subset ({\rm Sd}\, P_i)^{(1)}$ whose intersection type with $P_i$ is of the same type as that of $c$. A simplicial representative would not exist if the intersection type of $c$ and $P_i$ cannot be matched by that of a simplicial curve due to its required finiteness. 
The simplicial representative can 
be constructed as follows: First divide $c$ into connected components $c_\tau$ whose interior intersects the interior of only one cell $\tau \in P_i$. The segment can intersect many cells, but 
${\rm Int}(c_\tau) \cap {\rm Int}(\tau') \neq \emptyset \Rightarrow \tau = \tau'$. 
%%%
The $i$Rep of $c_\tau$ is the simplicial curve 
$c_{\tau, i}\subset {\rm Sd}P_i^{(1)}$ whose vertices are 
$\{ \sigma \in P_i \mbox{ such that } 
c_\tau \cap {\rm Int}(\sigma) \neq \emptyset \}$, where we consider that the interior of a 
cell is the cell minus its boundary, and 
our definition requires that the interior of a 
vertex be considered the vertex itself. 
The curve $c_i$ does not carry a parametrization; it is (the image by $\phi$ of) a simplicial curve: a collection of neighboring links and vertices. If $c$ has an orientation it is clear that $c_i$ inherits it. 

Similarly, the $i$Rep of a piecewise smooth surface $S \subset M$ (if it exists) is the only simplicial surface 
$S_i \subset ({\rm Sd}\, P_i)^{(2)}$ 
whose intersection type with $P_i$ is of the same type as that of $S$, and it can be constructed using the location of its vertices as done above. 
Also, an orientation in $S$ would induce one in $S_i$. 
Notice that if the surface $S$ has a boundary $(\partial S)_i = \partial S_i$. 
\end{dfn}

Now we will state the convergence result mentioned at the beginning of the section.  

\begin{tma}\label{tma:regularityandcurvature}
For every compact PL surface with boundary $S \subset M$ and 
every smooth \emph{closed} $2$-form $\omega$, we have
$$
\int_S\omega=\lim_{i\rightarrow\infty}\int_{S_i}\omega .
$$
In particular, 
for every smooth connection modulo gauge $[A]\in\AGp$ we have
\begin{equation*}
\int_S F^A = \lim_{i\rightarrow\infty} \int_{S_i}F^A = 
\lim_{i\rightarrow\infty} \sum_{\sigma \subset S_i} \omega_i^A(\sigma) .
\end{equation*}
\end{tma}

This convergence result holds because of the regularity properties of the sequence of scales that we defined. We will 
state some intermediate results in terms of areas measured according to an auxiliary piecewise smooth metric $g_0$ on $M$. 
Inside each of the simplices of ${\rm Sd}\, P_0$
we will use the euclidean metric in which the cells of $P_0$ are cartesian cubes of unit size. 

\begin{lma}\label{lma:regularidad1}
The boundary of the $i$Reps $S_i$ of any compact PL surface with boundary $S\subset M$ converge to $\partial S$ as the scale is refined in the following sense: 
\begin{enumerate}
\item\label{r1} 
Given two scales $P_i$, $P_j$ with $i \leq j$, call $A(i, j)$ the minimal area of a simplicial surface 
$S_{i, j}' \subset ({\rm Sd}\, P_j)^{(2)}$ such that 
$\partial S_{i, j}' = \partial S_j -  \partial S_i$; thus, for every 
$\epsilon >0$ there is a scale $P_k$ such that for every $j \geq i \geq k$ 

\[
A(i, j) < \epsilon . 
\]
\item\label{r2} 
Call $A(i)$ the minimal area of a PL surface $S_i'$ such that $\partial S_i' = \partial S -  \partial S_i$; then 
\[
\lim_{i \to \infty} A(i) = 0 . 
\]
\end{enumerate}
\end{lma}
\begin{proof}

In order to proceed with the proof we will recall some general properties of our cellular decompositions. 
The cellular decomposition $P_0$ can be constructed by gluing simplices of $({\rm Sd}|\Delta|, \phi)$. We consider 
$|{\rm Sd}\, \Delta|$ a subset of $\mathbb{R}^{\Delta}$ (one copy of $\mathbb{R}$ per simplex of the triangulation $(|\Delta|, \phi)$ of $M$), and we chose the metric on 
$\mathbb{R}^{\Delta}$ that is compatible with the metric $g_0$ making the cells of $P_0$ unit cubes. 

We maintain the same identification $\phi$ of $|\Delta|$ and $M$, and use it to study the triangulations ${\rm Sd}\, P_i$ by means of the induced refinements of 
$|{\rm Sd}\, \Delta| \subset \mathbb{R}^{\Delta}$. Notice that according to the metric $g_0$, ${\rm Sd}\, P_i$ is composed by isometric simplices. Moreover, 
there are $n$ types of $n$ dimensional simplices with the property that two simplices of the same type differ only by a rigid translation in $\mathbb{R}^{\Delta}$. 
We remark that the number of types of simplices of dimension $n$ is independent of $i$; it depends only on $n$, the dimension of the simplices, and on the number of $n$ dimensional cells in $P_0$. 

Now we start with the proof of the second point of the convergence statement. 
Given any PL surface $S\subset M$ consider a loop $l\subset M$ parametrizing one component of $\partial S$. 

For each $\sigma\in ({\rm Sd}\,P_0)^d$, we may assume that $l\cap\sigma$ consists of  a finite number of linear segments according to the euclidian structure of $\sigma$. Consider one of the linear segments $l'\subset l\cap\sigma$; 
its $i$Rep $l_i'\subset l_i$ consists of $n$ linear segments
$$
l'_i=l_{i,1}\cdot\dots\cdot l_{i,n},\qquad l_{i,1},\dots,l_{i,n}\in\left({\rm Sd}\,P_i\right)^{1} .
$$
Let $l_{i,k}^\bot$ be the normal euclidian projection of $l_{i,k}$ into $l'$. Notice that none of the segments $l_{i,k}$ is perpendicular to $l'$. Since there are finitely many lengths and directions available for $l_{i,k}$, 
the minimal length of the normal projections is greater than zero, 
$$
\min\{\mathrm{length}(l_{i, 1}^\bot),\dots,\mathrm{length}(l_{i,n}^\bot)\}>0 .
$$
Since the set of available directions for $l_{i,k}$ is independent of the scale $P_i$, 
we can conclude that 
ratio of the length of the segments $\mathrm{length}(l_{i,r})$, 
and the 
length of the normal projections $\mathrm{length}(l_{i,r}^\bot)$ 
is bounded by a finite number independent of the scale 
$$
M:=\max\left\{\frac{\mathrm{length}(l_{1,i})}{\mathrm{length}(l_{1,i}^\bot)},\dots,\frac{\mathrm{length}(l_{n,i})}{\mathrm{length}(l_{n,i}^\bot)}\right\} .
$$
Thus, 
$$
\mathrm{length}(l_i')=\sum_{k=1}^n\mathrm{length}(l_{i,k})\leq
M\sum_{k=1}^n\mathrm{length}(l_{i,k}^\bot)\leq M\cdot\mathrm{length}(l') . 
$$
Clearly, we can do the same for each component of $\partial S$; 
therefore, there is a constant $K>0$ independent of the scale such that 
\begin{equation*}
\mathrm{length}(\partial S_i)\leq K\cdot \mathrm{lenght}(\partial S) .
\end{equation*}

Hence, at scale $P_i$ there exists a polyhedral surface $S'_i\subset M$  with border 
$\partial S'_i= \partial S - \partial S_i$ (where we consider $\partial S , \partial S_i$ $1$-chains) such that
\begin{equation*}
	\mathrm{area}(S'_i)<(1+K)\cdot D_i\cdot \mathrm{length}(\partial S) ,
\end{equation*}
where $D_i$ is the diameter of the faces of $P_i$.

Since $\lim_{i\rightarrow\infty}D_i=0$, then 
$$
	\lim_{i \to \infty} \mathrm{area}(S'_i) = 0 .
$$
This concludes the proof of the second point of the lemma. 

The proof of the first point of the lemma uses the same ideas, but it is slightly more involved. 

For each connected component of $\partial S$ we have two $i$Reps, $l_i$ and $l_j$; the corresponding connected component of the simplicial surface that we are looking for, $S'_{i,j}$, then has $l_i - l_j$ as its border. Recall that $d =\dim M$, 
which means that each $d$ cube contains $d (d-1)$ $2$-faces; also recall that $P_i$ is coarser than $P_j$. With this information and the arguments above, we see that the area of this connected component is bounded by 
$$
(\mathrm{length}(l_i) + \mathrm{length}(l_j) ) \cdot  D_i \cdot  d (d-1) . 
$$
Thus, we conclude that 
\begin{equation}
\label{eqn:areaa2}
\mathrm{area}(S'_{i, j})<(2K) \cdot \mathrm{length}(\partial S) \cdot D_i \cdot  d (d-1)  .
\end{equation}
Again, since $D_i$ goes to zero as the scale is refined, this bound implies that 
the first statement of the lemma holds. 
\end{proof}

Now we will prove the convergence theorem of this section. 
\begin{proof}[Proof of Theorem \ref{tma:regularityandcurvature}]
Because of condition \ref{r2} of the previous lemma, at each scale $P_i$ there is an embedded PL surface $S_i'\subset M$, such that $\partial S_i'=\partial S - \partial S_i$ with 
$\lim_{i \to \infty} \mathrm{area}(S'_i) = 0$. 
Given that $S$ is compact, the subset of $M$ composed by the (closed) $d$-cells of $P_0$ which intersect $S'_i$ (for any $i$) is also compact. Thus, $\omega$ being smooth implies that 
$\int_{S_i'}\omega < C \cdot \mathrm{area}(S'_i)$, where $C$ is independent of $i$. 
Hence $\int_{S_i'}\omega$ converges to zero, and $\int_{S_i}\omega$ converges to $\int_S \omega$.
\end{proof}

%%%%%%%%%%%%%%%%%%%%%%%%%%%%%%%%%%

\section{Evaluation of the curvature in the continuum} \label{sec6}
Our framework is designed to deal with functions of the curvature of the connection; more specifically, we can talk about the integral of the curvature $2$-form on surfaces. This is the subject that we develop in the present section. 
The decomposition of the connection degrees of freedom given in Lemma \ref{lma:maps} 
and Theorem \ref{t1} implies that we can use measures that are factorized. In Section \ref{sec5} we presented 
measures which are gaussian in the $\underleftarrow{\Omega}_p$ factor (in {\em Case 1} and in {\em Case 2}); for details see the discussion following Theorem \ref{tma:Q}. 
When we calculate the expectation value of products of integrals of the curvature, the relevant part of the measure is the gaussian factor.

An oriented simplicial surface $S$ which fits in the $2$-skeleton, $({\rm Sd}\,P_i)^{(2)}$ induces a function of the curvature of the connection. 
We will call this function the $S$-curvature function. 
If the associated $2$-chain of $S$ is $\sum_{k=1}^nt_k\sigma^k_i$, then the continuous affine function $F_{S}:\underleftarrow{\Omega}_p\rightarrow\mathbb{R}$ is defined as
$$
	F_{S}	(\omega):=	\sum_{k=1}^n t_k \omega_i(\sigma_i^k)	
,\text{ where }\omega_i = \pi_i(\omega) .
$$
Notice that if $\omega^0$ restricted to $S$ is flat, then $F_S$ is $\omega^0$-homogeneous. 
Also, observe that by construction 
$$
F_{S}({\mathrm{Curv}}([A])) = \int_SF^A .
$$
We will describe how to extend the definition of $S$-curvature function to include \emph{any} oriented PL surface $S\subset M$. More precisely, we will define an affine function $F_S\in L^2(\underleftarrow{\Omega}_p,\rho)$. The idea is to consider 
$S_i\subset ({\rm Sd}\,P_i)^{(2)}$ at each scale, and prove $L^2$ convergence of the sequence $F_{S_i}$. In order to prove this convergence we will need 
the convergence result described in Point \ref{r1} of Lemma \ref{lma:regularidad1}. 

\begin{dfn}\label{dfn:new}
Let $\mathcal{H}_\rho\subset L^2\left(\underleftarrow{\Omega}_p,\rho\right)$ 
be the Hilbert subspace 
generated by all the $\omega^0$-homogeneous affine functions $f:{\Omega}_\rho\rightarrow\mathbb{R}.$ 
\end{dfn}

The following lemma characterizes the Hilbert space $\mathcal{H}_\rho$; see \cite{Ya} 9.6.

\begin{lma}\label{lma:new}
Consider the space $\underrightarrow{\Omega'}$ aided by the norm given by the positive bilinear form $\mathcal{Q}$ which is the covariance of $\rho$. 
The map 
$\mathcal{I}:\underrightarrow{\Omega'} \rightarrow \mathcal{H}_{\rho}$
defined by sending $\check{f} = \pi_i' \check{f_i}$ to the function 
$$
f(\omega)=f_i(\pi_i(\omega)) = f_i(\omega_i) = \check{f_i}(\pi_i(\omega) - \omega_i^0)
$$
is an isometry. 
\end{lma}

\begin{tma}
Let $\rho$ be the gaussian measure on $\underleftarrow{\Omega}_p$ centered on $\omega^0$ with covariance $\mathcal{Q}$ 
induced by a riemannian (or a polyhedral flat) metric in $M$ as in Theorem \ref{tma:Q}. 
For any oriented PL surface with boundary $S$ take the sequence of its $i$Reps, $S_i$. 
The sequence of affine functions $F_{S_i}:\underleftarrow{\Omega}_p\rightarrow\mathbb{R}$ 
converges to an affine function $F_S:{\Omega_\rho}\rightarrow\mathbb{R}$ 
in the sense of the $L^2$ norm induced by the measure $\rho$. 
\end{tma}

\begin{proof}
Consider linear functions $\check{F}_{S_i} \in \underrightarrow{\Omega'}$ defined by 
$\check{F}_{S_i}(\vec{\omega}):=F_{S_i}(\omega)-F_{S_i}(\omega^0)$. 
We claim that $\check{F}_{S_i}$ is a Cauchy sequence in $\underrightarrow{\Omega'}$ according to the norm 
$\|\cdot\|_{\mathcal{Q}}=\sqrt{\mathcal{Q}(\cdot,\cdot)}$. 
Consider a pair of scales $P_i$, $P_j$ with $j > i$. Notice that for every 
$\vec{\omega} \in\Omega^0_{j}$ we have
$$
\left(\check{F}_{S_i}-\check{F}_{S_j}\right)(\vec{\omega})=\check{F}_{S_{i,i+N}}(\vec{\omega}),
$$
where $S_{i,j}\subset ({\rm Sd}\,P_j)^{(2)}$ satisfies $\partial S_{i,j}=\partial S_i-\partial S_j$. 
Hence
$$
\|\check{F}_{S_i}-\check{F}_{S_j}\|_{\mathcal{Q}}=
\|\check{F}_{S_{i,j}}\|_{\mathcal{Q}}.
$$ 
Due to Point \ref{r1} of Lemma \ref{lma:regularidad1}, given any $\geps > 0$ there is a scale $P_k$ such that for every 
$j>i>k$ $area(S_{i,j}) < \geps$. 
Thus, it is a Cauchy sequence. 
Our previous lemma implies that the induced sequence $F_{S_i}$ in the Hilbert space $\mathcal{H}_\rho$ converges and defines an affine function 
$F_S:{\Omega_\rho}\rightarrow\mathbb{R}$. 
\end{proof}

The affine function $F_S:\Omega_\rho\rightarrow\mathbb{R}$ extends the $S$-curvature function of smooth connections in the sense specified by the following corollary of Theorem \ref{tma:regularityandcurvature}.
\begin{cor}\label{cor:F_S}
For any oriented PL surface with boundary $S\subset M$, 
and any connection modulo gauge $[A]\in\AGp$ we have 
\begin{equation*}
F_S(\omega^A):=\lim_{i\rightarrow\infty} F_{S_i}(\omega^A)=\int_SF^A 
\end{equation*}
where $\omega^A=\mathrm{Curv}([A]).$
\end{cor}

%%%%%%%%%%%%%%%%%%%%%%%%%%%%%%%%%%%%

\section{Independence of the choice of sequence of scales}\label{indep}
Now we will study how the constructed $S$-curvature functions depend on the choice of cubical sequence of scales $\mathscr{E}$. 
First we introduce the concept of a {\em regular sequence of scales}. 
This broader category of sequences of scales is more cumbersome, but it is large enough to let us compare two different cubical sequences of scales. 

%%%

\begin{dfn}[Regular sequence of scales]
We consider 
cellular decompositions of $M$ denoted by $(P_i, \phi_i)$, or by $P_i$ to simplify notation, which have convex closed polyhedral cells. 
We add the further requirement that the cells of $P_i$ are piecewise linear according to the PL structure determined by $(|\Delta |, \phi)$. 
For each cellular decomposition $P_i$ we also consider a triangulation of $M$ which refines $P_i$ by adding a vertex at the interior of each cell, 
$t_i : |{\rm Sd}\,P_i | \rightarrow M$. Here by $|{\rm Sd}\,P_i|$ we mean the simplicial complex which arises from considering the set of cells of 
$P_i$ as its vertex set following the procedure described in Definition \ref{Sd}, but the location of the new vertices is not the baricenter of the corresponding cell 
according to the affine structure of the cells of $P_i$. Instead, the location of the new vertices is set by the map $t_i$ which is the restriction of an affine map 
when restricted to each cell of $|{\rm Sd}\,P_i|$ onto the corresponding cell of $P_i$. 
Thus, we have enriched the structure of the cellular decompositions, and we 
consider triples $(P_i, |{\rm Sd}\,P_i |, t_i)$, 
but when the context allows it we will denote them simply by $P_i$.  Notice that the provided structure allows us to discuss the $i$Reps of curves and 
surfaces with boundary as in the previous sections. 

We say that a sequence 
$$
\mathscr{E}=\{(P_i,\bar{\phi}_i)\,\mid\,i\in\mathbb{N}\}
$$ 
is a sequence of scales for $M$ if the following conditions are satisfied: 
\begin{enumerate}
\item 
$P_j$ is a refinement of $P_j$ of $P_i$ for every $j>i$.
\item 
The triangulation $t_j : |{\rm Sd}\,P_j | \rightarrow M$ refines $t_i : |{\rm Sd}\,P_i | \rightarrow M$  for every $j>i$.
\item 
The sequence is exhaustive, i.e. for each open set $U\subset M$ there is a scale $P_i$ with a $d= \dim M$ dimensional 
simplex $\sigma\in(|{\rm Sd}\,P_i |^d, t_i)$ with $\sigma\subset U$, $d=\dim M$.
\end{enumerate}

The sequence of scales is said to be {\bf a regular sequence of scales} if the boundary of the $i$-Reps $S_i$ of any compact PL surface with boundary $S \subset M$ 
converge to $\partial S$ 
as the scale is refined in the same sense as in Lemma \ref{lma:regularidad1}. We use the same 
auxiliary piecewise smooth metric $g_0$ on $M$ used in  Lemma \ref{lma:regularidad1}. 
To implement this idea we add the following conditions: 
\begin{enumerate}
\setcounter{enumi}{3}
\item\label{R1} 
Given two scales $P_i$, $P_j$ with $i \leq j$, call $A(i, j)$ the minimal area of a simplicial surface 
$S_{i, j}' \subset ({\rm Sd}\, P_j)^{(2)}$ such that 
$\partial S_{i, j}' = \partial S_j -  \partial S_i$; for every 
$\epsilon >0$ there is then a scale $P_k$ such that for every $j \geq i \geq k$ 
\[
A(i, j) < \epsilon . 
\]
\item\label{R2} 
Call $A(i)$ the minimal area of a PL surface $S_i'$ such that $\partial S_i' = \partial S -  \partial S_i$; then 
\[
\lim_{i \to \infty} A(i) = 0 . 
\]
\end{enumerate}
\end{dfn}

%%%
Notice that the cubical sequences of scales described in Section \ref{convergence} are regular sequences of scales.

Let $\mathscr{E}=\{P_i \}$ and $\mathscr{E}'=\{P_i'\}$ be two sequences of regular scales; 
since $M$ is compact and both sequences of scales are PL according to the same PL structure, we will be able to construct a new sequence of scales 
$\mathscr{E}''=\{P_i''\}$ (which may not be regular) whose elements $(P_i'', |{\rm Sd}\,P_i'' |, t_i'')$ refine corresponding elements of $\mathscr{E}$ and $\mathscr{E}'$ in the obvious sense. This sequence of scales will be a key ingredient in proving the results that we present in this section. 

Given the scales $(P_0, |{\rm Sd}\,P_0 |, t_0)$ and $(P'_0, |{\rm Sd}\,P'_0 |, t'_0)$ there is a triangulation $(|\Delta_0''|, \psi_0)$ which refines 
triangulations $(|{\rm Sd}\,P_0 |, t_0)$ and $(|{\rm Sd}\,P'_0 |, t'_0)$. Clearly, it also refines the cellular decompositions $P_0$ and $P'_0$. 
Now observe that 
for any choice of $t''_0$ the triangulation $|{\rm Sd}\,\Delta_0''| ; t''_0)$, being a refinement of $(|\Delta_0''|, \psi_0)$, is also a refinement of 
triangulations $(|{\rm Sd}\,P_0 |, t_0)$ and $(|{\rm Sd}\,P'_0 |, t'_0)$. 
Thus, the first element of the new sequence of scales can be considered to be $(|\Delta_0''|, \psi_0; |{\rm Sd}\,\Delta_0''| ; t''_0)$. 

We proceed to construct the $i$th element of the new sequence of scales for $i>0$. 
Given the scales $(P_i, |{\rm Sd}\,P_i |, t_i)$ and $(P'_i, |{\rm Sd}\,P'_i |, t'_i)$ there is a triangulation $(|\Delta_i''|, \psi_i)$ which refines 
triangulations $(|{\rm Sd}\,P_0 |, t_0)$, $(|{\rm Sd}\,P'_0 |, t'_0)$ and $(|\Delta_{i-1}''|, \psi_{i-1})$. Notice that it also refines $P_i$ and $P'_i$. 
Now we use the same arguments as above: we see that the triangulation $|{\rm Sd}\,\Delta_i''| ; t''_i)$ is a refinement of 
triangulations $(|{\rm Sd}\,P_i |, t_i)$ and $(|{\rm Sd}\,P'_0 |, t'_0)$ for any $t''_i$. If we adjust $t''_i$, the triangulation 
$(|{\rm Sd}\,P_i |, t_i)$ also refines $(|{\rm Sd}\,P_{i-1} |, t_{i-1})$ which means that $\mathscr{E}''=\{ (|\Delta_i''|, \psi_i; |{\rm Sd}\,\Delta_i''| ; t''_i) \}$ 
is a sequence of scales.

Thus there are three projective limits of vector spaces
$$
	\underleftarrow{\Omega}_{p,{\mathscr E}},\qquad
	\underleftarrow{\Omega}_{p,{\mathscr E}''},\qquad	
	\underleftarrow{\Omega}_{p,{\mathscr E}'} . 
$$
By the common refinement property of the elements of 
$\mathscr{E}''$ we have affine maps 
$\boldsymbol{\varpi}_i:{\Omega}_{i,p,\mathscr{E}''}\rightarrow
{\Omega}_{i,p,\mathscr{E}}$, and 
$\boldsymbol{\varpi}_i':{\Omega}_{i,p,\mathscr{E}''}\rightarrow
{\Omega}_{i,p,\mathscr{E}'}$, such that the following diagrams commute
$$\xymatrix{
 &
 \Omega_{i,p,\mathscr{E}}
 &
 \Omega_{j,p,\mathscr{E}}\ar[l]_{\ti{\pi}_{ji}}
\\
\Omega_{i,p,\mathscr{E}''}
\ar[ur]^{\boldsymbol{\varpi}_i}\ar[dr]^{\boldsymbol{\varpi}_i'}
&
\Omega_{j,p,\mathscr{E}''}
\ar[l]_{\ti{\pi''}_{ji}}
\ar[ur]^{\boldsymbol{\varpi}_j}\ar[dr]^{\boldsymbol{\varpi}_j'}
&\\
&
\Omega_{i,p,\mathscr{E}'}
&
\Omega_{j,p,\mathscr{E}'}\ar[l]^{\ti{\pi'}_{ji}}
}.$$
Therefore, taking the projective limit we have the maps
$$\xymatrix{
&\underleftarrow{\Omega}_{p,{\mathscr E}}\\
\underleftarrow{\Omega}_{p,{\mathscr E}''}
\ar[ur]^{\boldsymbol{\varpi}}\ar[dr]^{\boldsymbol{\varpi}'}\\
&\underleftarrow{\Omega}_{p,{\mathscr E}'}
}.$$

Notice the curvature maps 
$\mathrm{Curv}_{\mathscr{E}}: \AGp\rightarrow\underleftarrow{\Omega}_{p,{\mathscr E}}$, 
$\mathrm{Curv}_{\mathscr{E}'}: \AGp\rightarrow\underleftarrow{\Omega}_{p,{\mathscr E}'}$ 
and $\mathrm{Curv}_{\mathscr{E}''} : \AGp\rightarrow\underleftarrow{\Omega}_{p,{\mathscr E}''}$ 
are defined, and they are compatible with the structures described above 
in the sense that 
$$
\boldsymbol{\varpi} \circ \mathrm{Curv}_{\mathscr{E}''} = \mathrm{Curv}_{\mathscr{E}} \, \, , \, \, 
\boldsymbol{\varpi}' \circ \mathrm{Curv}_{\mathscr{E}''} = \mathrm{Curv}_{\mathscr{E}'} .
$$

Let $\rho,\rho'$ and $\rho''$ be gaussian probability measures on $\underleftarrow{\Omega}_{p,{\mathscr E}}$, $\underleftarrow{\Omega}_{p,{\mathscr E}'}, $ and  $\underleftarrow{\Omega}_{p,{\mathscr E}''}$, centered on $\omega^0,(\omega^0)'$ and $(\omega^0)''$ respectively, with $\varpi((\omega^0)'')=\omega^0$ and $\varpi'((\omega^0)'')=(\omega^0)'$. 
The covariance of the measures is defined 
by the same area form on $M$ as in Theorem \ref{tma:Q}. We will denote by $\Omega_\rho$, $\Omega_{\rho'}$ and $\Omega_{\rho''}$ the supports of $\rho,\rho'$ and $\rho''$ respectively. These measure spaces are related as stated in the following lemma. 
\begin{lma}\label{lma:isomorphism}
The 
$(\omega^0)''$-$(\omega^0)$ homogeneous affine 
map $\boldsymbol{\varpi}$ and 
the 
$(\omega^0)''$-$(\omega^0)'$ homogeneous affine 
map $\boldsymbol{\varpi}'$ 
induce affine isomorphisms 
from $\Omega_{\rho''}$ to 
$\Omega_\rho$ and $\Omega_{\rho'}$ respectively. 
Moreover, the isomorphism 
$$
\boldsymbol{\varpi}\circ(\boldsymbol{\varpi}')^{-1} : \Omega_{\rho'} \to \Omega_{\rho}
$$
is independent of our choice of $\mathscr{E}''$. 
\end{lma}
\begin{proof}
First notice that since $\boldsymbol{\varpi}_*\rho''=\rho$ then $\boldsymbol{\varpi}(\Omega_{\rho''})\subset\Omega_\rho$, and $\varpi:\Omega_{\rho''}\rightarrow\Omega_{\rho}$ is a surjective linear map.

Let us prove that it is also an injective map. Let $\mathcal{H}_\rho,\mathcal{H}_{\rho'},\mathcal{H}_{\rho''}$ be the Hilbert spaces given as in Definition \ref{dfn:new} 
for the measures $\rho,\rho',\rho''$ respectively.
Consider the dual isometries
$$\xymatrix{
&\underrightarrow{\Omega'}_\mathscr{E}
\ar[dl]^{(\boldsymbol{\varpi})^*}
\\
\underrightarrow{\Omega'}_{\mathscr{E}''}
\\
&\underrightarrow{\Omega'}_{\mathscr{E}'}
\ar[ul]^{(\boldsymbol{\varpi}')^*}
} . $$
We claim that 
the isometry $(\boldsymbol{\varpi})^*$ can be extended to an isomorphism between the Hilbert spaces 
$\mathcal{H}_\rho$,  $\mathcal{H}_{\rho''}$. 
In order to prove it we need to see that 
the image of $(\boldsymbol{\varpi})^*\left(\underrightarrow{\Omega'}_\mathscr{E}\right)\subset\underrightarrow{\Omega'}_{\mathscr{E}''}$ under the corresponding isometric inclusion of Lemma \ref{lma:new}, $\mathcal{I}:\underrightarrow{\Omega'}_{\mathscr{E}''}\rightarrow \mathcal{H}_{\rho''}$, is a dense subspace of $\mathcal{H}_{\rho''}$.
Notice that 
since $\boldsymbol{\varpi}^*$ is an isometry, 
for any $2$-simplex $\sigma\in ({\rm Sd}\,|\Delta_j''|^2,\phi_j'')$ the affine function
$F_{{\sigma}}:\underleftarrow{\Omega}_{p,{\mathscr E}''}\rightarrow\mathbb{R}$ 
is an accumulation point of the sequence $\{ F_{\sigma_i}\circ\boldsymbol{\varpi} \}_\mathscr{E}$ 
induced by the $\mathscr{E}$ $i$Reps of $\sigma$. 
Thus, Lemma \ref{lma:new} implies that our claim holds.

Now we come back to the proof of injectivity of the map $\boldsymbol{\varpi}:\Omega_{\rho''}\rightarrow\Omega_\rho$. Suppose that there exists a one dimensional affine subspace $L\subset \Omega_{\rho''}\subset \underleftarrow{\Omega}_{p,{\mathscr E}''}$ that lies in the kernel of 
%the $(\omega^0)''$-$(\omega^0)$ homogeneous affine map 
$\boldsymbol{\varpi}:\underleftarrow{\Omega}_{p,{\mathscr E}''}
\rightarrow\underleftarrow{\Omega}_{p,{\mathscr E}}$. 
Let $\vec{\omega}\in \underleftarrow{\Omega}^0_{\mathscr{E}''}$ be a generator of the linear subspace $L-\omega^0 \subset \underleftarrow{\Omega}^0_{\mathscr{E}''}$. 
Since $L\subset \Omega_{\rho''}$, 
Lemma \ref{lma:new} implies that 
there is $f \in \underrightarrow{\Omega'}_{\mathscr{E}''}$ with $\|f\|_{\mathcal{Q}''}\neq 0$ such that 
$f(\vec{\omega}) \neq 0$. 

On the other hand, 
$L$ being in the kernel of 
$\boldsymbol{\varpi}:\underleftarrow{\Omega}_{p,{\mathscr E}''}
\rightarrow\underleftarrow{\Omega}_{p,{\mathscr E}}$ 
means that if the sequence of linear functions $\boldsymbol{\varpi}^*(h_n) \in \underrightarrow{\Omega'}_{\mathscr{E}''}$ 
converges to $h$ contained in the 
$\|\cdot\|_{\mathcal{Q}''}$-closure of $\underrightarrow{\Omega'}_{\mathscr{E}''}$, 
then $h(\vec{\omega})=0$. 
Therefore $f$ cannot belong to the closure of $\boldsymbol{ \varpi}^*(\underrightarrow{\Omega'}_{\mathscr{E}})$, which contradicts the fact that 
$\mathcal{I}\left(\boldsymbol{\varpi}^*(\underrightarrow{\Omega'}_{\mathscr{E}})\right)$ is a dense subspace of ${\mathcal{H}}_{\rho''}$. 

Now we will prove that 
$\boldsymbol{\varpi}\circ(\boldsymbol{\varpi}')^{-1} : \Omega_{\rho'} \to \Omega_{\rho}$ 
is independent of our choice of $\mathscr{E}''$. 
Compatibility of the maps $\boldsymbol{\varpi}$, $\boldsymbol{\varpi}'$, $\boldsymbol{\varpi}''$ with the curvature maps implies 
$$
\boldsymbol{\varpi}\circ(\boldsymbol{\varpi}')^{-1} \circ \mathrm{Curv}_{\mathscr{E}'} = \mathrm{Curv}_{\mathscr{E}} . 
$$
This observation, and the denseness of 
$\mathrm{Curv}_{\mathscr{E}} (\AGp ) \subset \underleftarrow{\Omega}_{p,{\mathscr E}}$, 
$\mathrm{Curv}_{\mathscr{E}'} (\AGp ) \subset \underleftarrow{\Omega}_{p,{\mathscr E}'}$ 
imply that 
$\boldsymbol{\varpi}\circ(\boldsymbol{\varpi}')^{-1} : \Omega_{\rho'} \to \Omega_{\rho}$ 
is independent of our choice of $\mathscr{E}''$. 
\end{proof}

Notice that the notion of regular sequence of scales lets us define 
$S$-curvature functions
for any 
compact oriented PL-surface with boundary $S\subset M$. 
The $S$-curvature functions induced by different sequences of regular scales are related
\begin{tma}
Given any compact oriented PL-surface with boundary $S\subset M$, and any 
two regular sequences of scales 
${\mathscr{E}}$, ${\mathscr{E}'}$, 
the induced 
$S$-curvature functions are equivalent 
in the sense that 
$$
F_{S,\mathscr{E}} = 
F_{S,\mathscr{E}'} \circ \boldsymbol{\varpi}\circ(\boldsymbol{\varpi}')^{-1} . 
$$
\end{tma}
\begin{proof}
The property 
$$
\boldsymbol{\varpi}\circ(\boldsymbol{\varpi}')^{-1} \circ \mathrm{Curv}_{\mathscr{E}'} = \mathrm{Curv}_{\mathscr{E}} ,
$$
together with the denseness of 
$\mathrm{Curv}_{\mathscr{E}} (\AGp ) \subset \underleftarrow{\Omega}_{p,{\mathscr E}}$ and 
$\mathrm{Curv}_{\mathscr{E}'} (\AGp ) \subset \underleftarrow{\Omega}_{p,{\mathscr E}'}$  
imply that
the desired equation holds. 
\end{proof}

%%%%%%%%%%%%%%%%%%%%%%%%%%

\section{Outlook: Construction of physical theories}
Here we briefly describe how our work can be used to construct  
physical theories. 
An euclidean quantum field theory could be constructed following 
an implementation of Wilson's renormalization group: 
We described the space of histories for abelian gauge theories 
as the extended space of curvature evaluations, 
and we gave a family of 
kinematical gaussian measures $\rho$ on it. In order to construct 
an euclidean quantum field theory the starting point is a 
sequence of measures $\mu_i$ on $\Omega_{i,p}$ defined by an effective theory at 
scale $i$. At each scale the physical measure is continuous with respect to the Lebesge measure and with respect to the gaussian measure; i.e. we can write 
$d\mu_i=W_i(\omega_i, \lambda_i) d{\mathsf Leb}=W'_i(\omega_i, \lambda_i) d\rho_i$. 
As usual, the coupling constant (or constants) $\lambda_i$ is determined by 
a renormalization prescription; that is, adjusting the constant(s) to reach 
agreement with a chosen macroscopic ``experiment'' (or experiments). 

The resulting sequence of measures $d\mu_i$ may not be compatible with coarse graining; in fact, one would expect that effective theories at finer scales be better even when modeling the behavior of observables of macroscopic character. Thus 
if $i \leq j$, the coarse graining of $d\mu_j$ to scale $i$ is expected to correct $d\mu_i$; we will call this corrected measure $d\mu_{i(j)}$. If the corrected measures $d\mu_{i(j)}$ converge as $j \to \infty$ we say that the resulting measure is completely renormalized, 
$d\mu_i^{\mathsf ren}= \lim_{j \to \infty} d\mu_{i(j)}$. 

As described above the completely renormalized measure can only yield predictions for observables which are functions on one of the spaces $\Omega_{i,p}$, but as we saw in previous sections of the article, our framework let us study other observables using a regularization procedure. An observable can be regularized at each scale resulting in a function $f_i : \Omega_{i,p} \to \mathbb{R}$ and we may ask for the convergence of 
$\lim_{i \to \infty} \mu_i(f_i)$. This limit and the limit written in the previous paragraph are related (see \cite{MOWZ}). 
The procedure followed in Section \ref{sec6} let us treat a family of observables which is independent of the sequence of scales used to define the continuum limit. Then we can ask if the resulting continuum limit enjoys of the same symmetries as the classical theory, or if the breaking of the classical symmetry group by the introduction of an arbitrary sequence of measuring scales persist after the continuum limit is performed. 
The analysis done in Section \ref{indep} could be seen as a warm up for dealing with recovery of symmetry through the continuum limit for the case of a physically relevant measure. 
If our field theory is defined on Minkowski space, we should study the validity of the Osterwalder-Schrader axioms; in particular, the issue of symmetry recovery at the continuum limit is one of the points that must be addressed. 
If we are attempting to construct a theory in which there is no metric background, like in quantum gravity, we would work along lines similar to those pioneered in \cite{AMMT} to finalize the construction of a quantum gauge field theory. 

%%%%%
\section*{Acknowledgements}
JZ acknowledges partial support from CONACyT grant 80118.

%%%%%%%%%%%%%%%%%%%%%%%%%%%%%%%%%%%%%%%%%%%%%%%%%%%%%%%%%%%%%%%%%%%%%%%

\end{document}